\documentclass[11pt,letterpaper]{article}
\usepackage[margin=1in]{geometry}
\usepackage{color}
\usepackage{graphicx}
\usepackage{amssymb,amsmath,amsthm}
\usepackage{mathtools}
\usepackage{array}
\usepackage{lineno}
\usepackage{algorithm,algpseudocode}
\usepackage{hyperref}
\usepackage{authblk}

\newtheorem{lemma}[]{Lemma}
\newtheorem{theorem}[]{Theorem}

\newtheorem{problem}{Problem}

\DeclareMathOperator{\LF}{LF}
\DeclareMathOperator{\BWT}{BWT}
\DeclareMathOperator{\SA}{SA}
\DeclareMathOperator{\LCP}{LCP}
\DeclareMathOperator{\LCE}{LCE}
\DeclareMathOperator{\ISA}{ISA}
\DeclareMathOperator{\DSA}{DSA}
\DeclareMathOperator*{\argmin}{arg\,min}

\title{Compressibility-Aware Quantum Algorithms on Strings}

\author[1]{Daniel Gibney \thanks{daniel.j.gibney@gmail.com}} 
\author[2]{Sharma V.~Thankachan \thanks{sharma.thankachan@gmail.com}}

\affil[1]{Georgia Institute of Technology, Atlanta}
\affil[2]{North Carolina State University, Raleigh}

\begin{document}

\maketitle

\begin{abstract}

Sublinear time quantum algorithms have been established for many fundamental problems on strings. 
This work demonstrates that new, faster quantum algorithms can be designed when the string is highly compressible. 
We focus on two popular and theoretically significant compression algorithms --- the Lempel-Ziv77 algorithm (LZ77) and the Run-length-encoded Burrows-Wheeler Transform (RL-BWT), and obtain the  results below.  

We first provide a quantum algorithm running in $\tilde{O}(\sqrt{zn})$ time for finding the LZ77 factorization of an input string $T[1..n]$ with $z$ factors. Combined with multiple existing results, this yields an $\tilde{O}(\sqrt{rn})$ time quantum algorithm for finding the RL-BWT encoding with $r$ BWT runs. Note that $r = \tilde{\Theta}(z)$.
We complement these results with lower bounds proving that our algorithms are optimal (up to polylog factors). 

Next, we study the problem of compressed indexing, where 
we provide a $\tilde{O}(\sqrt{rn})$ time quantum algorithm for constructing a recently designed $\tilde{O}(r)$ space structure with equivalent capabilities as the suffix tree. This data structure is then applied to numerous problems to obtain  sublinear time quantum algorithms when the input is highly compressible.
For example, we show that the longest common substring of two strings of total length $n$ can be computed in $\tilde{O}(\sqrt{zn})$ time, where $z$ is the number of factors in the LZ77 factorization of their concatenation. 
This beats the best known $\tilde{O}(n^\frac{2}{3})$ time quantum algorithm when $z$ is sufficiently small. 

\end{abstract}

\newpage 

\section{Introduction}
\label{sec:intro}
Algorithms on strings (or texts) are at the heart of many applications in Computer Science. Some of the most fundamental problems include finding the longest repeating substring of a given string, finding the longest common substring of two strings, checking if a given (short) string appears as a substring of another (long) string, sorting a collection of strings, etc. All of these problems can be solved in linear time using a classical computer; here and hereafter, we assume that our alphabet (denoted by $\Sigma$) is a set of polynomially-sized integers. 
This linear time complexity is typically necessary, and hence optimal, as the entire input must be at least read for most problems. However, for quantum computers, many of these problems can be solved in sub-linear time.

One of the most powerful data structures in the field of string algorithms is the suffix tree. It can be constructed in optimal, linear time, after which it facilitates solving hundreds of important problems. The suffix tree takes $O(n)$ words of space, where $n$ is the length of the input string. A recent result by Gagie et al.~shows that the suffix tree can be encoded in $O(r\log(n/r))$ space, and still supports all the functionalities of its uncompressed counterpart with just logarithmic slowdown~\cite{DBLP:journals/jacm/GagieNP20}. 
Here $r$ denotes the 
the number of runs in its Burrows-Wheeler Transform (BWT). The value $r$ has been shown to be within polylogarithmic factors of other popular measures of compressibility including the number of factors of the LZ77 factorization of the string $z$, and the $\delta$-measure -- a foundational lower bound on string's compressibility~\cite{DBLP:journals/cacm/KempaK22}. Specifically, $\delta < z, r = O(\delta \log^2 n)$ and they can be orders of magnitude smaller than $n$ for highly compressible strings. The bounds relating these compressibility measures and the existence of compressed indexes  motivate and enable our work.
Our main results are summarized below: 

\begin{itemize}
\item We provide a quantum algorithm running in $\tilde{O}(\sqrt{zn})$ time for finding the LZ77 factorization of the input string. This combined with multiple existing results, yields an $\tilde{O}(\sqrt{rn})$ time quantum algorithm for finding the run-length encoded BWT of the input string. To the best of our knowledge, these are the first set of quantum results on these problems.

\item We complement the above results with lower bounds proving that our algorithms are optimal (up to polylog factors). These are shown to hold for binary alphabets. For larger alphabets, the same lower bound holds for even computing the value $z$. 

\item We provide an $\tilde{O}(\sqrt{rn})$ time quantum algorithm for constructing an $\tilde{O}(r)$ space structure with equivalent capabilities as the suffix tree. Many fundamental problems can now be easily solved by applying either classical algorithms or quantum subroutines. Examples include: 

\begin{itemize}

\item Finding the longest common substring (LCS) between two strings of total length $n$ in $\tilde{O}(\sqrt{zn})$ time, where $z$ is the number of factors in the LZ77 parse of their concatenation. For highily compressible strings, this beats the best known $\tilde{O}(n^\frac{2}{3})$ time quantum algorithm~\cite{DBLP:conf/soda/AkmalJ22,DBLP:conf/innovations/GallS22}. We can also find the set of maximal unique matches (MUMs) between them in $\tilde{O}(\sqrt{zn})$ time. Similar bounds can be obtained for finding the longest repeating substring and shortest unique substring of a given string.


\item Obtaining the Lyndon factorization of a string in $\tilde{O}(\sqrt{\ell n})$ time where $\ell$ is the number of its Lyndon factors. To the best of our knowledge, this is the first proposed quantum algorithm for solving this problem (in time sub-linear when $\ell = n^{1-\Theta(1)}$). 

\item Determining the frequencies of all distinct substrings of length $q$ ($q$-grams) in time $\tilde{O}(\sqrt{zn} + occ)$  where $occ$ is the number of distinct $q$-grams.
\end{itemize}

\end{itemize}
Our results indicate that string compressibility can be exploited to design new, faster quantum algorithms for many problems on strings.

\subsection{Related Work}

Previous quantum algorithms for string problems include an $\tilde{O}(\sqrt{n} + \sqrt{m})$ algorithm for determining if a pattern $P[1..m]$ is a substring in another string $T[1..n]$~\cite{DBLP:journals/jda/HariharanV03} (see for \cite{ablayev2020quantum} for further optimizations), an $\tilde{O}(n^\frac{2}{3})$ algorithm for finding the longest common substring of two strings~\cite{DBLP:conf/soda/AkmalJ22}, and an $\tilde{O}(\sqrt{n})$ time algorithm for longest palindrome substring~\cite{DBLP:conf/innovations/GallS22}. Quantum speedups have also been found for the problems of sorting collections of strings~\cite{khadiev2022quantum}, pattern matching in non-sparse labeled graphs~\cite{darbari,equi2023bit}, computing string synchronizing sets and performing pattern matching with mismatches~\cite{DBLP:journals/corr/abs-2211-15945}.


Khadiev and Remidovskii~\cite{DBLP:journals/nc/KhadievR21} considered the problem of reconstructing of a text $T[1..n]$ from a fixed set of small strings $\mathcal{S} = \{S_1$,$\hdots$, $S_m\}$ of total length $L$. A solution to the reconstruction problem should output a $S_1'$, $\hdots$, $S_k'$ where $S_j' \in \mathcal{S}$ and positions $q_1$, $\hdots$, $q_k$ such that $q_j \leq q_{j-1} + |S_j'|$ and $T[q_j.. q_j + |S_j'|-1] = S_j'$ for $1\leq j \leq k$ and $q_1 = 1$, $q_k = n - |S_k'| + 1$. They present a quantum algorithm running in time $\tilde{O}(n + \sqrt{mL})$. Note that this problem is fundamentally different from finding the LZ77 encoding as the set of small strings $\mathcal{S}$ is fixed and independent of the text for this reconstruction problem and not for the problem of finding the LZ77 encoding.

The problem of determining a string from a set of substring queries of the form `is $S$ in $T$' was considered by Cleve et al.~\cite{DBLP:conf/swat/CleveIGNTTY12} who demonstrated a $\Omega(n)$ query complexity lower bound for that problem. This, again, is fundamentally different from the problem considered here, in that the symbols in $T$ at a particular index cannot be directly queried, making the problem harder than the one considered in this paper.

\subsection{Roadmap}
Section \ref{sec:prelimin} provides the necessary background and technical preliminaries. After providing these, in Section \ref{sec:alg} we present the algorithms for obtaining the LZ77 and RL-BWT encodings of the input text. In Section \ref{sec:SA_index}, we show how to construct the suffix array index. We then go on to apply this data structure with classical and quantum algorithms in Section \ref{sec:applications_} to numerous problems to demonstrate its utility. We then provide computational lower bounds in Section \ref{sec:lb} and close with some open problems in Section \ref{sec:discussion}.

\section{Preliminaries}
\label{sec:prelimin}

\subsection{LZ77 Compression}
LZ77 compression was introduced by Ziv and Lempel~\cite{DBLP:journals/tit/ZivL77} and works by dividing the input text $T[1..n]$ into factors or disjoint substrings so that each factor is either the leftmost occurrence of a symbol or an instance of a substring that exists further left in $T$. There are variations on this scheme, but we will use the form primarily studied in the survey by Narvarro~\cite{DBLP:journals/csur/Navarro21} and used by Kempa and Kociumaka in \cite{DBLP:journals/cacm/KempaK22}. The following greedy algorithm can find the LZ77 factorization: 
Initialize $i$ to $1$, and repeat the following three steps:
\begin{enumerate}
    \item If $T[i]$ is the first occurrence of a given alphabet symbol then we make $T[i]$ the next factor and $i \gets i + 1$.
    \item Otherwise, find the largest $j \geq i$ such that $T[i..j]$ occurs in $T$ starting at a position $i' < i$. Make $T[i..j]$ the next factor and then make $i \gets j+1$.
    \item If $i \leq n$, continue, otherwise end.
\end{enumerate}
Overall possible compression methods based on textual substitutions that only use self-reference to factors starting earlier in the string, the LZ77 greedy strategy creates the smallest number of factors~\cite{DBLP:journals/jacm/StorerS82}. Based on the factorization, the string can be encoded by having every factor represented by either a new symbol, if it is the first occurrence of that symbol, or by recording the start position of the earlier occurrence, say $p_i$, and the length of the factor, $\ell_i$. 
The naive algorithm for computing the LZ77 encoding requires $O(n^2)$ time; however, the factorization can be computed in linear time by modifying linear time suffix tree construction algorithms~\cite{DBLP:journals/jacm/RodehPE81}.
An example of this factorization with factors separated by $\mid$, and the corresponding encoding is in Figure \ref{fig:lz77_example}. To denote the $i^{th}$ factor, we use $F_i = (s_i, \ell_i)$ where $s_i$ is the starting position of the factor, and $\ell_i$ is the length of the factor. The text can now be encoded as $(p_1, \ell_1),(p_2, \ell_2),(p_3, \ell_3), \dots, (p_z, \ell_z)$, where $p_i = T[s_i]$ (a symbol in $\Sigma$) if $s_i$ is the first occurrence of $T[s_i]$, otherwise $p_i < s_i$ is an index, where $T[p_i..p_i+\ell_i-1] = T[s_i..s_i+\ell_i-1]$. 

\subsection{LZ-End Compression}

LZ-End was introduced by Kreft and Narvarro~\cite{DBLP:journals/tcs/KreftN13} to speed up the extraction of substrings relative to traditional LZ77. Unlike LZ77, LZ-End forces any new factor that is not a new symbol to end at a previous factor boundary, i.e., $T[s_i..s_i+\ell_i-1]$ is taken as the largest string possible that is a suffix of $T[1..s_j-1]$ for some $j \leq i$. An example is given in Figure \ref{fig:lz77_end_example}.
Like LZ77, LZ-End can be computed in linear time~\cite{DBLP:conf/dcc/KempaK17,DBLP:conf/esa/KempaK17}. Moreover, the LZ-End encoding size is close to the size of LZ77 encoding, as shown in the following recent result by Kempa and Saha. This result will be generalized here and is instrumental in how we obtain our quantum algorithm.

\begin{lemma}[\cite{DBLP:conf/soda/KempaS22}]
\label{lem:z_e_and_z}
For any string $T[1..n]$ with LZ77 factorization size $z$ and LZ-End factorization size $z_e$, we have $z_e = O(z \log^2 n)$.
\end{lemma}

\begin{figure}
\resizebox{\textwidth}{!}{%
    \centering
    \begin{tabular}{m{12mm} m{7mm}| m{7mm} | m{7mm} | m{7mm} | m{7mm} m{7mm} | m{7mm} m{7mm} m{7mm} m{7mm} m{7mm} | m{7mm} m{7mm} m{7mm} |m{7mm}}
Index &1 & 2 & 3 & 4 & 5 & 6 & 7 & 8 & 9 & 10 & 11 & 12 & 13 & 14 & 15\\
$T$ & $a$  & $b$ & $a$  & $c$ & $a$ & $b$ & $c$ & $a$ & $b$ & $c$ & $a$ & $a$ & $a$ & $a$ & $b$ \\
\end{tabular}
}
\caption{The LZ77 factorization and encoding of $T = abacabcabcaaaab$ with $n = 15$ and $z = 8$.}
\label{fig:lz77_example}
\end{figure}

\begin{figure}
\resizebox{\textwidth}{!}{%
    \centering
    \begin{tabular}{m{12mm} m{7mm}| m{7mm} | m{7mm} | m{7mm} | m{7mm} m{7mm} | m{7mm} m{7mm} m{7mm} | m{7mm} | m{7mm} | m{7mm} | m{7mm} m{7mm} |m{7mm}}
Index &1 & 2 & 3 & 4 & 5 & 6 & 7 & 8 & 9 & 10 & 11 & 12 & 13 & 14 & 15\\
$T$ & $a$  & $b$ & $a$  & $c$ & $a$ & $b$ & $c$ & $a$ & $b$ & $c$ & $a$ & $a$ & $a$ & $a$ & $b$ \\
\end{tabular}
}
\caption{The LZ-End factorization and encoding of $T = abacabcabcaaaab$ with $z_{e} = 11$.}
\label{fig:lz77_end_example}
\end{figure}

\subsection{Relationship Between LZ77 and other Compression Methods}
The relationship between multiple compressibility measures has only recently become quite well understood. The $\delta$-measure (also known as substring complexity) is defined as $\delta = \max_{1\leq i \leq n} \frac{d_i}{i}$, where $d_i$ is the number of distinct substrings of length $i$ in $T$ and was originally introduced by Raskhodnikova et al.~\cite{DBLP:journals/algorithmica/RaskhodnikovaRRS13} to provide a sublinear time approximation algorithm for the value of $z$.
The $\delta$-measure provides a lower bound on other compression measures such as the size $\gamma$ of the smallest string attractor~\cite{DBLP:conf/stoc/KempaP18}, the smallest bi-directional macro scheme~\cite{DBLP:journals/jacm/StorerS82}, whose size is denoted $b$, the number of rules $g$ in the smallest context-free grammar generating the string~\cite{DBLP:journals/tit/KiefferY00}, and the number of runs $r$ in the Burrows-Wheeler Transform (BWT) of the string~\cite{burrows1994block}.
We refer the reader to surveys by Navarro~\cite{DBLP:journals/csur/Navarro21a,DBLP:journals/csur/Navarro21} for descriptions of these compression measures and their applications.
Importantly, it was shown by Kempa and Kociumaka~\cite{DBLP:journals/cacm/KempaK22} that $r = O(\delta \log^2 n)$, which combined with $\delta \leq z$ implies $r = O(z \log^2 n)$. Because $r$ bounds $b$ and $\gamma$ from above~\cite{DBLP:journals/csur/Navarro21a} and we can obtain a grammar of size $\tilde{O}(z)$ generating the text using the LZ77 encoding. Therefore, these forms of compression all have encodings that are now provably within logarithmic factors from one another in terms of size.

\subsection{Suffix Trees, Suffix Arrays, and the Burrow Wheeler Transform}

\begin{figure}
    \centering
    \begin{minipage}{.75\textwidth}
    \includegraphics[width=\textwidth]{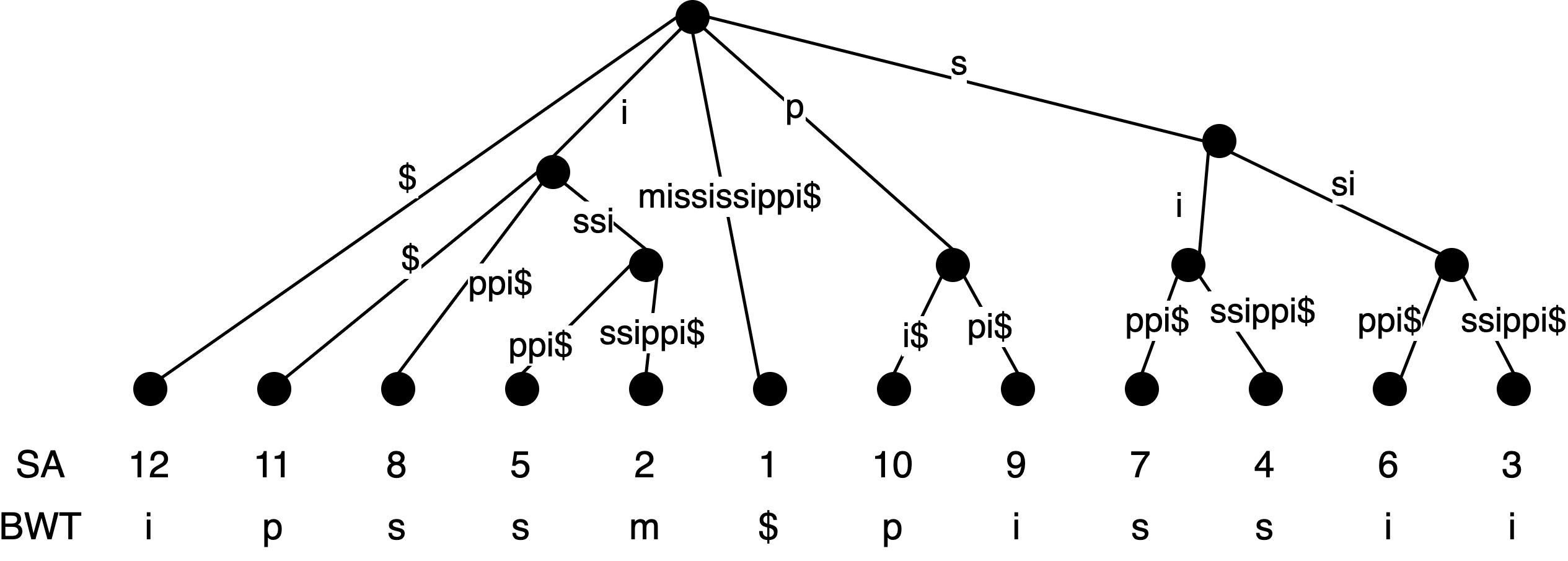}
    \end{minipage}
    \hspace{1em}
    \begin{minipage}{.2\textwidth}
    \includegraphics[width=\textwidth]{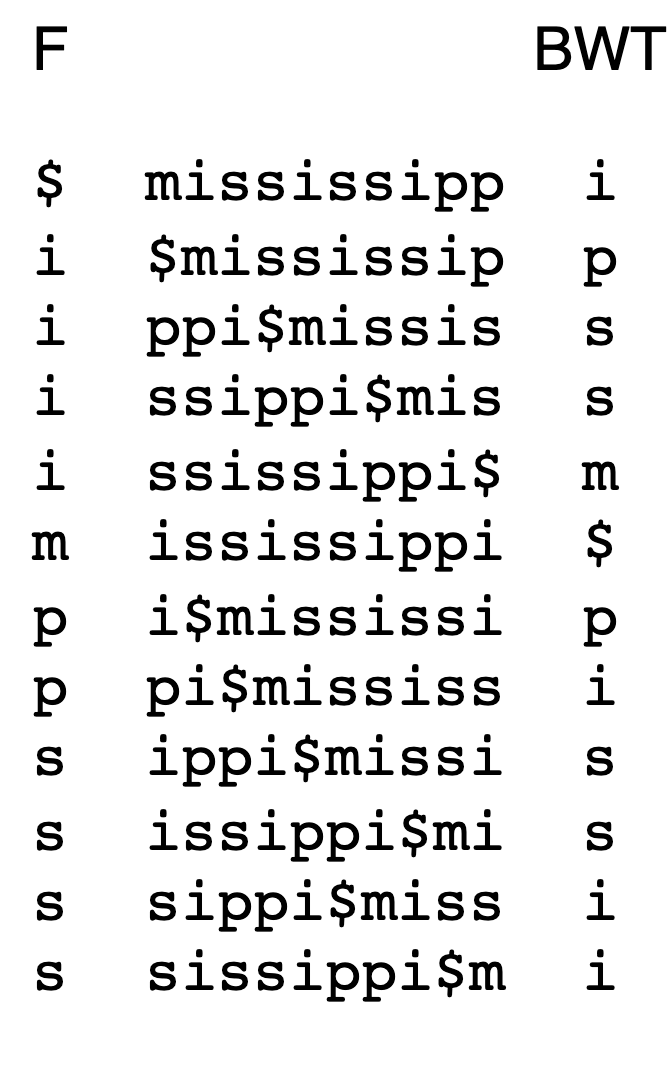}
    \end{minipage}
    \caption{(Left) The suffix tree, suffix array, and BWT of the string mississippi\$. (Right) The matrix of sorted cyclic shifts of mississippi\$.}
    \label{fig:ST_BWT}
\end{figure}

We assume that the last symbol in $T[1..n]$ is a special $\$$ symbol that occurs only at $T[n]$ and is lexicographically smaller than the other symbols in $T$.
The suffix tree of a string $T$ is a compact trie constructed from all suffixes of $T$. The tree leaves  are labeled with the starting index of the corresponding suffix and are sorted by the lexicographic order of the suffix. These values in this order define the suffix array $\SA$, i.e., $\SA[i]$ is such the $T[\SA[i]..n]$ is the $i^{th}$ largest suffix lexicographically. See Figure \ref{fig:ST_BWT} (left).  The inverse suffix array $\ISA$, is defined as $\ISA[\SA[i]] = i$, or equivalently, $\ISA[i]$ is equal to lexicographic rank of the suffix $T[i..n]$. The Burrows Wheeler Transform (BWT) of a text $T$ is a permutation of the symbols of $T$ such that $\BWT[i] = T[\SA[i]-1]$. This BWT is frequently illustrated using the last column of the matrix constructed by sorted all cyclic shifts of $T$. The first column is labeled $F$. See Figure \ref{fig:ST_BWT} (right).
The longest common extension of two suffixes $T[i..n]$, $T[j..n]$ is denoted as $\LCE(i,j)$ and is equal to the length of their longest common prefix. 
Suffix trees, suffix arrays, and the Burrows-Wheeler Transform can all be computed 
in linear time for polynomially-sized integer alphabets~\cite{DBLP:conf/focs/Farach97}.
While suffix trees/arrays require space $O(n)$ words (equivalently $O(n\log n)$ bits), the BWT requires only $n \log|\Sigma|$ bits. Further, we can apply run-length compression on BWT to achieve $O(r \log n)$ bits of space. 

\subsection{FM-index and Repetition-Aware Suffix Trees}
The FM-index provides the ability to count and locate occurrences of a given pattern efficiently. 
It is constructed based on the BWT described previously and uses the \emph{LF-mapping} to perform pattern matching. The LF-mapping is defined as $\LF[i] = \ISA[\SA[i]-1]$, that is, $\LF[i]$ is the index in the suffix array corresponding to the suffix $T[\SA[i]-1..n]$. The FM-index was developed by Ferragina and Manzini~\cite{DBLP:conf/focs/FerraginaM00} to be more space efficient than traditional suffix trees and suffix arrays. However, supporting location queries utilized sampling $\SA$ in evenly spaced intervals, in a way independent of the runs in the BWT of the text, preventing a $O(r)$ data space structure with optimal (or near optimal) query time.

The r-index and subsequent fully functional text indexes were developed to utilize only $O(r)$, or $\tilde{O}(r)$ space. The r-index developed by Gagie et al.~\cite{DBLP:conf/soda/GagieNP18} was designed to occupy $O(r)$ space and support pattern location queries in near-optimal time. It was based on the observation that suffix array samples are necessary only for the run boundaries of the BWT and subsequent non-boundary suffix array values can be obtained in polylogarithmic time. The fully functional indexes given by Gagie et al.~\cite{DBLP:journals/jacm/GagieNP20} use $\tilde{O}(r)$ space and provide all of the capabilities of a suffix tree. The data structure allows one to determine in $\tilde{O}(1)$ time arbitrary $\SA$, $\ISA$ and $\LCE$ values, which in turn lets one determine properties of arbitrary nodes in suffix tree, such as subtree size. We will show how to construct this index after obtaining our compressed representation in Section \ref{sec:SA_index}.

\subsection{Quantum Computing and QRAM}
We assume that our quantum algorithm can access the symbol $T[i]$ by querying the input oracle with the query `$i$'. The \emph{query complexity} of a quantum algorithm is the number of times the input oracle is queried. The time complexity of a quantum algorithm is measured in terms of the number of elementary gates\footnote{A definition of elementary gates can be found in \cite{barenco1995elementary}.} used to implement the algorithm with a quantum circuit in addition to the number of queries made to the input oracle. 
This implies that the algorithm's query complexity always lower bounds the time complexity of a quantum algorithm.
Under the assumption of quantum random access~\cite{DBLP:journals/corr/abs-2203-05599}, classical algorithms can be invoked by a quantum algorithm with only constant factor overhead. Specifically, a classical algorithm running in $T(n)$ time can be implemented with $O(T(n))$ time~\cite{DBLP:conf/soda/AkmalJ22}. We also assume that we have classical control over our algorithm, which evokes quantum subroutines based on Grover's search (described  below). Specifically, we assume that our algorithm can be halted based on the current computational results returned from a quantum sub-routine. The output of a quantum algorithm is probabilistic; hence we say a quantum algorithm solves a problem if it outputs the correct solution with probability at least $\frac{2}{3}$. 

One of the most fundamental building blocks for quantum algorithms is Grover's search~\cite{DBLP:conf/stoc/Grover96}. The most basic version of Grover's search allows one, given an  oracle $f: \{1,\hdots, n\} \rightarrow \{0,1\}$ where there exists a single $i \in \{1, \hdots, n\}$ such that $f(i) = 1$ to determine this $i$ in $O(\sqrt{n} \cdot t)$ queries and time where $t$ is the time for evaluating $f$ on a given index. Variants of Grover's search can also be used to determine if there exists such an $i$ in $O(\sqrt{n} \cdot t)$ queries or even find the index $i$ such that $f(i)$ is minimized~\cite{DBLP:journals/corr/quant-ph-9607014}, as we will use next.

\section{Quantum Algorithms for Compression}
\label{sec:alg}

Before presenting our main results, we first present a basic building block based on Grover's search that will be used throughout this section.

\paragraph{The Rightmost Mismatch Algorithm.} A variant of Grover's search designed by Durr and Hoyer~\cite{DBLP:journals/corr/quant-ph-9607014} allows one to solve the following problem in $O(\sqrt{n})$ time and queries with probability at least $\frac{2}{3}$: Given an oracle $f:\{1,\hdots,n\} \rightarrow \{1,\hdots,m\}$, determine an index $i \in \{1,\hdots,n\}$ such that $f(i)$ is minimized.  This minimum finding algorithm can be applied so that for any two substrings of $T$ of length $\ell$, say $T[i-\ell+1..i]$ and $T[j-\ell+1..j]$, it returns the smallest $k \in \{1, \hdots, \ell\}$ such that  $T[i-k+1] \neq T[j-k+1]$. In particular, given the oracle for $T$, we define the oracle $f_{ij}^\ell:\{1,\hdots,\ell\} \rightarrow \{1,\hdots,\ell+1\}$ as 
\[
f_{ij}^\ell(k) = \begin{cases}
k & \text{if } T[i-k+1] \neq T[j-k+1]\\
\ell+1 & \text{if } T[i-k+1] = T[j-k+1].
\end{cases}
\]
and apply the minimum finding algorithm to it. The rightmost mismatch algorithm uses $O(\sqrt{\ell})$ time and queries on substrings of length $\ell$. Using this we can identify whether two substrings of length $\ell$ match, and if not, their \emph{co-lexicographic order} (lexicographic order of the reversed strings) by comparing their rightmost mismatched symbol.

\subsection{Algorithms with Near Optimal Query Complexity}

This section provides two potential solutions with optimal and near-optimal query times. The first has optimal $O(\sqrt{zn})$ query complexity but requires exponential time. The second has a near-optimal $\tilde{O}(\sqrt{zn})$ query complexity but requires $\tilde{O}(n)$ time. The second algorithm introduces ideas that will be expanded on in Section \ref{sec:sublinear_time_alg} for our main algorithm, where the query and time complexity are  with $\tilde{O}(\sqrt{zn})$. 

\subsubsection{Achieving Optimal-Query Complexity in Exponential Time}
\label{sec:oracle_id_algorithm}


A naive approach is to first obtain the input string from the oracle (in the worst case using $n$ oracle queries).
Then any compressed representation can be computed without further input queries. The first approach discussed here shows how to find the input string using fewer queries, specifically $O(\sqrt{zn})$ queries for binary strings. We will prove this query complexity is optimal in Section \ref{sec:lb}. This algorithm is based on a solution for the problem of identifying an oracle (in our case, an input string) in the minimum number of oracle queries  by Kothari~\cite{DBLP:conf/stacs/Kothari14}. 
Kothari's solution builds on a previous `halving' algorithm by Littlestone~\cite{DBLP:journals/ml/Littlestone87}. 

We next describe the basic halving algorithm as applied to our problem.
Assuming that $z$ is known, we enumerate all binary strings of length $n$ with at most $z$ LZ77 factors. Call this set $\mathcal{S}$. Since, an encoding with $z$ factors encoding requires at most $2z\log n$ bits, there are at most $\sum_{i=0}^{2z\log n} 2^i =  2^{2z\log n + 1} - 1 = 2n^{2z} - 1$ such strings in $\mathcal{S}$. We construct a string $M$  of length $n$ from $\mathcal{S}$ where $M[i] = 1$ if at least half of strings in $\mathcal{S}$ are $1$ at the $i^{th}$ position, and $M[i] = 0$ otherwise. Note that the construction of $M$ requires time exponential in $z$ but does not require any oracle queries. Grover's search is then used to find a mismatch if one exists between $M$ and the oracle string with $O(\sqrt{n})$ queries. If a mismatch occurs at position $i$, we can then eliminate at least half of the potential strings in $\mathcal{S}$. We repeat this process until no mismatches are found, at which point we have completely recovered the oracle (input string). Known algorithms can then obtain all compressed forms of text. 

Naively applying this approach would   result in an algorithm with $\log |\mathcal{S}| \cdot \sqrt{n} = O(z \sqrt{n}\log n)$ query complexity. 
Kothari's improvements on this basic halving algorithm give us a quantum algorithm that uses $\sqrt{n \log |\mathcal{S}| / \log n} = O(\sqrt{zn})$ input queries. We can avoid assuming knowledge of $z$, by progressively trying different powers of $2$ as our guess of $z$, still resulting in $\sum_{i = 0}^{\log z} \sqrt{2^in} = O(\sqrt{zn})$ queries overall. As noted above, this approach is not time efficient.


\subsubsection{Achieving Near-Optimal Query Complexity in Near-Linear Time}
\label{sec:linear_time}

An algorithm with similar query complexity and far improved time complexity is possible by using a more specialized algorithm. Specifically, one can obtain the non-overlapping LZ77 factorization. For non-overlapping LZ77, every factor, say $T[s_i..s_i+\ell_i-1]$, that is not new symbol must reference a previous occurrence completely contained in $T[1..s_i-1]$. This only increases the  size of this factorization by at most a logarithmic factor. That is, if  $z_{no}$ is the number of factors for the non-overlapping LZ77 factorization, then $z \leq z_{no}=O(z\log n)$~\cite{DBLP:journals/csur/Navarro21a}. This factorization can be converted into other compressed forms in near-linear time, as described in Section \ref{sec:other_encodings}.

We obtain the factorization by processing $T$ from left to right as follows: Suppose inductively that we have the factors $F_1 = (s_1, \ell_1)$, $F_2 = (s_2, \ell_2)$, $\hdots$, $F_{i-1} = (s_{i-1}, \ell_{i-1})$ in (start, length) encoding and want to obtain the $i^{th}$ factor. 
Assume that we have the prefixes $T[1..k]$, for $k \in [1, s_i-1]$ sorted in co-lexicographic order. To find the next factor $T[s_i.. j]$ we apply exponential search\footnote{Recall that exponential search checks ascending powers of $2$ until an interval $[2^{x-1}, 2^x]$ for some $x \geq 1$ containing the solution is found, at which point binary search is applied to the interval.} on $j$. To evaluate a given $j$ we use binary search on the sorted set of prefixes. To compare a prefix $T[1..k]$, we use the rightmost mismatch algorithm on the substrings $T[s_i.. j]$ and $T[k-(j-s_i).. k]$. If no rightmost mismatch is found, then $T[s_i.. j]$ has occurred previously as a substring and we continue the exponential search on $j$. Otherwise, we compare the symbol at the rightmost mismatch to identify which half of the sorted set of prefixes to continue the binary search. If $\ell_i$ is the length of $i^{th}$ factor found, this requires $O(\log^2 n \cdot \sqrt{\ell_i})$ queries and time. 

To proceed to the $(i+1)^{th}$ factor, we now must obtain the co-lexicographically sorted order of the $\ell_i$ new prefixes. This can be done using a standard linear time suffix tree construction algorithm. Specifically, if we consider the suffix tree of the reversed text $T^R$, we are prepending $\ell_i$ symbols to a suffix of $T^R$. These are accessed from either the oracle directly only in the case the new factor is a new symbol, and otherwise from the previously obtained string. Since we are prepending to $T^R$, a right-to-left suffix construction algorithm such as McCreight's~\cite{DBLP:journals/jacm/McCreight76} can be used. 

A final issue to be addressed is that every call to the rightmost mismatch algorithm only returns the correct result with probability at least $\frac{2}{3}$. As is standard, we will repeat each call to the rightmost mismatch algorithm some number $t$ times and take the majority solution as the answer, or, if there is no majority, any of the most frequent solutions. The probability that our final answer is incorrect
is at most  $e^{-2(\frac{2}{3} - \frac{1}{2})^2t} = e^{-\frac{1}{18}t}$ (See Lemma 2.2 in \cite{valiant2019cs265}). Additionally, the probability that our algorithm as a whole is incorrect is bound by the probability that one or more of these majorities taken are incorrect. By union bound, this is at most the sum of the probabilities that these individual majorities are incorrect. This is the sum of at most $2z \cdot \log^2 n \leq  2n \cdot \log^2 n$ probabilities, since each factor is associated with at most $2\log n$ `combined $t$ calls' (counting the combined $t$ calls with a majority taken as one) for exponential search, and each of those with $\log n$ `combined $t$ calls' from binary search. Hence the sum of the probabilities, or the probability that our algorithm uses a wrong solution at any point, is bound by $2n \log^2 n \cdot e^{-\frac{1}{18}t}$. To make $2n\log^2 n \cdot e^{-\frac{1}{18}t} < \frac{1}{3}$, it suffices to make $t =  \lceil 18\ln (6n\log^2 n) \rceil = O(\log n)$.

With probability at least $\frac{2}{3}$, the query complexity is
$\sum_{i=1}^{z_{no}} \sqrt{\ell_i} \cdot t \log^2 n = \tilde{O}(\sum_{i=1}^{z_{no}} \sqrt{\ell_i}).
$ At the same time, we have $\sum_{i=1}^{z_{no}} \ell_i = n$, so the sum is maximized when each $\ell_i = \frac{n}{z_{no}}$ making $\sum_{i=1}^z \sqrt{\ell_i} \leq \sqrt{z_{no} n}$.
Hence, the query complexity is $\tilde{O}(\sqrt{zn})$. The time complexity is $O(n + \log^3 n \cdot \sqrt{z_{no}n})$, which is $\tilde{O}(n)$. We will focus for the rest of this section on developing these ideas and utilizing more complex data structures to obtain a sublinear time algorithm.

\subsection{Main Algorithm: Optimal Query and Time Complexity}
\label{sec:sublinear_time_alg}

\subsubsection{High-Level Overview}

On a high level, the algorithm will proceed very much like the near-linear time algorithm from Section \ref{sec:linear_time}. It proceeds from left to right, finding the next factor and utilizes a co-lexicographically sorted set of prefixes of $T$. After the next factor is found, a set of new prefixes of $T$ is added to this sorted set. However, we face two major obstacles: (i) we cannot afford to explicitly maintain a sorted order of all prefixes needed to  check all possible previous substrings efficiently; (ii) if we utilize a factorization other than LZ77, like LZ77-End where fewer potential positions have to be checked, then the monotonicity of being a next factor is lost, i.e., if for LZ77-End $T[s_i..j]$ may have  occurred as a substring ending at a previous factor, but $T[s_i.. j-1]$ may not have  occurred as a substring ending at a previous factor.

To overcome these problems, we introduce a new factorization scheme that extends the LZ-End factorization scheme discussed in Section \ref{sec:prelimin}. It allows for more potential places ending locations for each new factor obtained by the algorithm.

\subsubsection{LZ-End+$\tau$ Factorization} Let $\tau \geq 1$ be an integer parameter. The LZ-End+$\tau$ factorization of the string $T$ from left to right. Initially $i \gets 1$. For $i \geq 1$, if $T[i]$ is not in $T[1..i-1]$ we make $T[i]$ a new factor and make $i \gets i+1$. Otherwise, let $j$ be the largest index such that $T[i..j]$ has an early occurrence ending at either the last position in an earlier factor or at a position $k < i$ such that $k-1 \equiv 0 \mod \tau$ (the $-1$ is to account for not 0-indexing on strings). Let $z_{e+\tau}$ denote the number of factors created by the LZ-End+$\tau$ factorization.

Note that there exist strings where $z_e < z_{e+\tau}$.
The smallest binary string example where this is true is $00010011011$, which has an LZ-End factorization with seven factors $0$, $0$, $0$, $1$, $001$, $1$, $011$ 
and an LZ-End+$\tau$ for $\tau = 2$ with eight factors $0$, $0$, $0$, $1$, $001$, $10$, $1$, $1$. 
Loosely speaking, the LZ77-End+$\tau$ algorithm can be `tricked' into taking a longer factor earlier on, in this case the factor `$10$' which is possible for LZ77-End+$\tau$ but not LZ77-End, and limits future choices.
Fortunately, the same bounds in terms of $z$  established by Kempa and Saha~\cite{DBLP:conf/soda/KempaS22} hold for $z_{e+\tau}$ that hold for $z_e$. 

\begin{lemma}
\label{lem:LZ-End-tau}
Let $z_{e+\tau}$ (resp., $z$) denote the number of factors in the LZ-End+$\tau$ (resp., LZ77) factorization of a given text $T[1..n]$. Then it holds that $z_{e+\tau} = O(z\log^2 n)$.
\end{lemma}

\begin{proof}
See Appendix \ref{appendix:lz-end_proof}.  
\end{proof}




Next, we describe how new LZ77-End+$\tau$ factors of $T$ are obtained by using the concept of the \emph{$\tau$-far property} and a dynamic longest common extension (LCE) data structure. Following this, we describe how the co-lexicographically sorted prefixes required by the algorithm are maintained.

\subsubsection{Maintaining the Colexicographic Ordering of Prefixes}

\begin{figure}
    \centering
    \includegraphics[width=.55\textwidth]{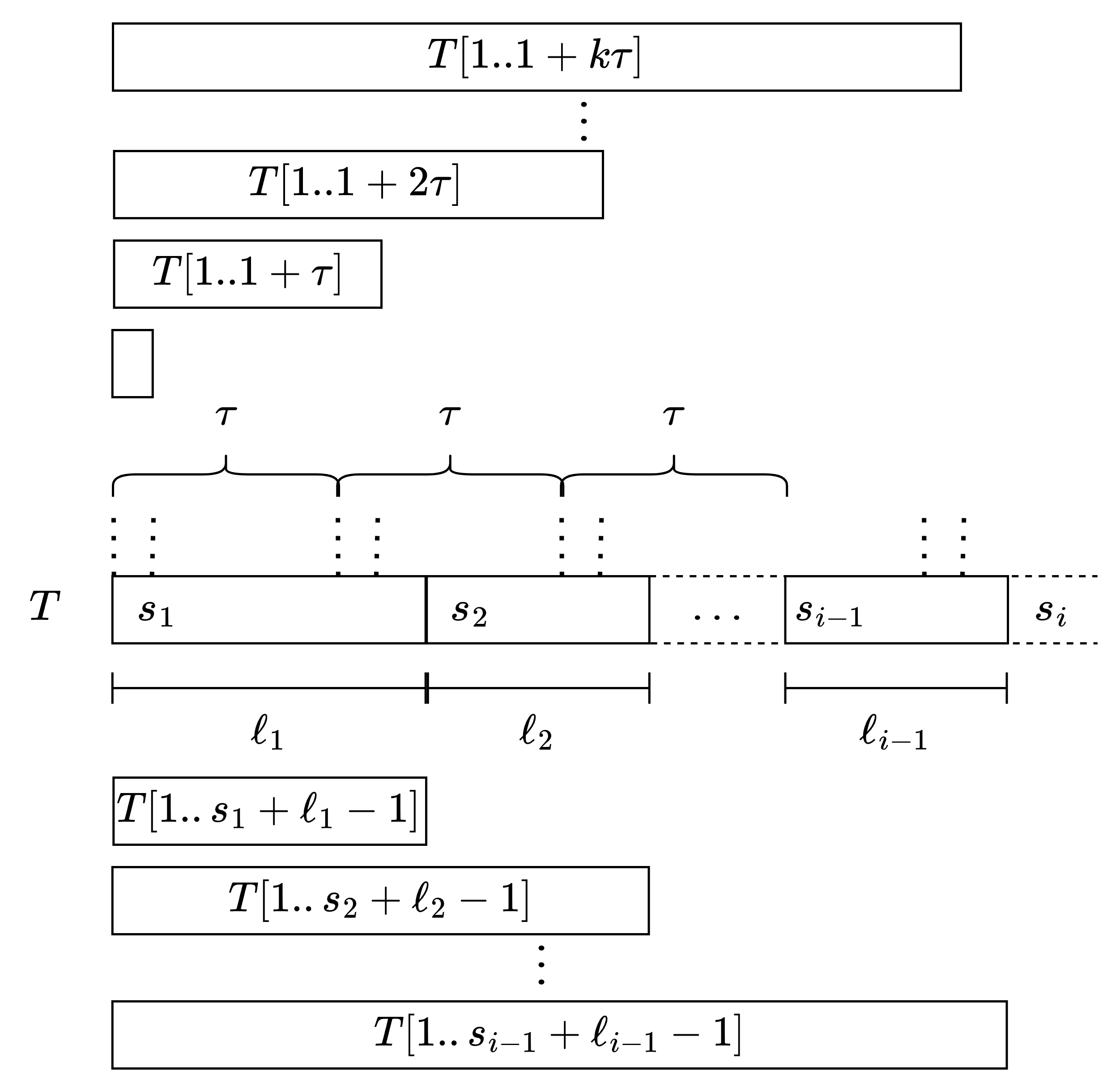}
    \caption{The colexicographic order of the prefixes  $T[1..s_1 + \ell_1 -1]$, $T[1..s_2 + \ell_2-1]$, $\hdots$, $T[1..s_{i-1} + \ell_{i-1}-1]$ (shown below $T$) and $T[1..1]$, $T[1..\tau+1]$, ..., $T[1.. k\tau+1]$ where $k$ is the largest natural number such that $1 + k\tau  < s_i$ (shown above $T$) are known prior to iteration $i$.}
    \label{fig:factors}
\end{figure}

The first factor $F_1$ is always of the form $(T[1])$. Assume inductively that the factors $F_1$, $F_2$, ..., $F_{i-1}$ have already been determined. Recall that for factors $F_j$ of the form $(\alpha)$, $\alpha \in \Sigma$ we also store $(s_j, 1)$ where $s_j$ is the starting position of the $j^{th}$ factor in $T$. We assume inductively that we have the colexicographically sorted order of prefixes of 
\begin{align*}
\mathcal{P}_{i-1} &\coloneqq \{T[1.. s_j + \ell_j - 1] \mid  (s_j, \ell_j) = F_j, 1 \leq j \leq i-1\}\\ 
&\cup \{T[1.. j] \mid 1\leq j \leq s_{i-1} + \ell_{i-1}-1,~~ j-1 \equiv 0 \mod \tau\}.
\end{align*}
See Figure \ref{fig:factors} for an illustration of the prefixes contained in $\mathcal{P}_i$. For each of these we store the ending position of the prefix.

The following section will show how to obtain the factor $F_i$. 
For now, suppose we just determined the $i^{th}$ factor starting position $s_i$.
After the factor length $\ell_i$ is found, we need to determine where to insert the prefixes $T[1.. s_i + \ell_i - 1]$ and $T[1..j]$ for $j \in [s_i, s_i + \ell_1 - 1]$ where $j-1 \equiv 0 \mod \tau$, in the colexicographically sorted order of $\mathcal{P}_{i-1}$  to create $\mathcal{P}_i$. To do this, we use the dynamic longest common extension (LCE) data structure of Nishimoto et al.~\cite{DBLP:conf/mfcs/NishimotoIIBT16} (see Lemma \ref{lem:dynamic_lce}). 



\begin{lemma}[Dynamic LCE data structure~\cite{DBLP:conf/mfcs/NishimotoIIBT16}]
\label{lem:dynamic_lce}
An LCE query on a text $S[1..m]$ consists of two indices $i$ and $j$ and returns the largest $\ell$ such that $S[i..i+\ell] = S[j..j+\ell]$. 
There exists a data structure that requires $O(m)$ time to construct, supports LCE queries in $\tilde{O}(1)$ time, and supports insertion of either a substring of $S$ or single character into $S$ at an arbitrary position in $\tilde{O}(1)$ time\footnote{Polylogarithmic factors here are with respect to final string length after all insertions.}. 

\end{lemma}
The main idea is to use the above dynamic LCE structure over the 
reverse of the prefix of $T$ found thus far. 
We initialize the dynamic LCE data structure with the first LZ-End+$\tau$ factor of $T$, which is a single character. 
For every factor found after that, we prepend the reverse factor to the current reversed prefix and update the data structure, all in $\tilde{O}(1)$ time. In particular, if the $i^{th}$ factor of $T$ found is a new character, we prepend that character to our dynamic LCE structure for
$S \coloneqq (T[1..s_i-1])^R$ (the reverse of $T[1..s_i-1])$.
If the $i^{th}$ factor found is $T[s_i.. s_i + \ell_i - 1] = T[x..y]$, for $x,y \in [1, s_{i}-1]$, then we prepend the substring $(T[x..y])^R = S[s_i-y.. s_i-x]$
to string representation of our dynamic LCE structure.
Once the reversed $i^{th}$ factor is prepended to the reverse prefix in the dynamic LCE  structure, to compare the colexicographic order of the new prefixes in $\mathcal{P}_i$ we find the LCE of the two reversed prefixes being compared and compare the symbol in the  position after their furthest match. Applying this comparison technique and binary search on $\mathcal{P}_{i-1}$, we determine where each  prefix in $\mathcal{P}_i \setminus \mathcal{P}_{i-1}$ should be inserted in the sorted order in logarithmic time.

\subsubsection{Finding the Next LZ-End+$\tau$ Factor}
\label{sec:next_factor}

We now  show how to obtain the new factor $F_i = (s_i,\ell_i)$.
Firstly, $s_i = s_{i-1} + \ell_{i-1}$. 
We say $T[s_i..h]$ is a \emph{potential factor} if 
either $h=s_i$ and $T[s_i..h] \in \Sigma $ is the leftmost occurrence of a symbol in $T$, 
or $T[s_i..h] = T[x..y]$ where $y < s_i$ and $y$ is end of the previous factor or $y-1 \equiv 0 \mod \tau$.
We say the \emph{$\tau$-far property} holds for an index $j \geq s_i$ if there exists $h$ such that $j-\tau \leq h \leq j$ and $T[s_i..h]$ is a potential factor.

\begin{figure}
    \centering
    \begin{minipage}{.6\textwidth}
    \includegraphics[width=\textwidth]{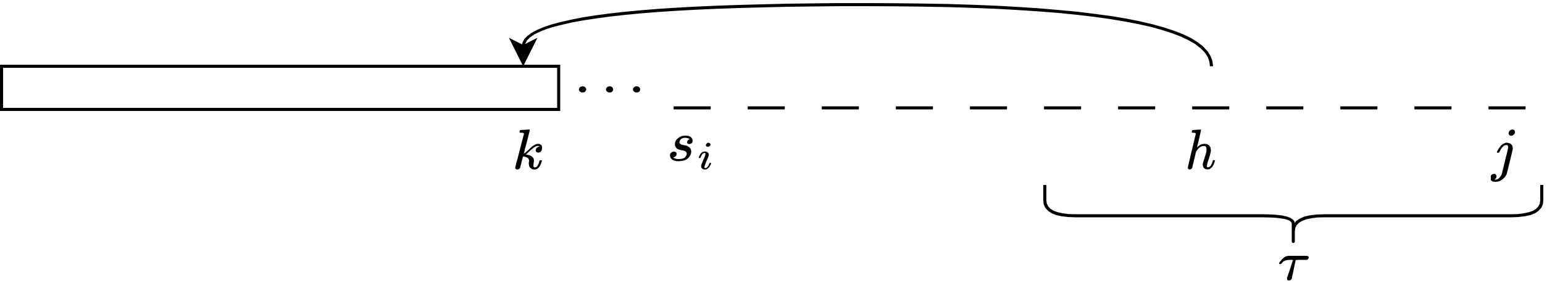}
    \end{minipage}
    
    \vspace{1em}
    \begin{minipage}{.75\textwidth}
    \includegraphics[width=\textwidth]{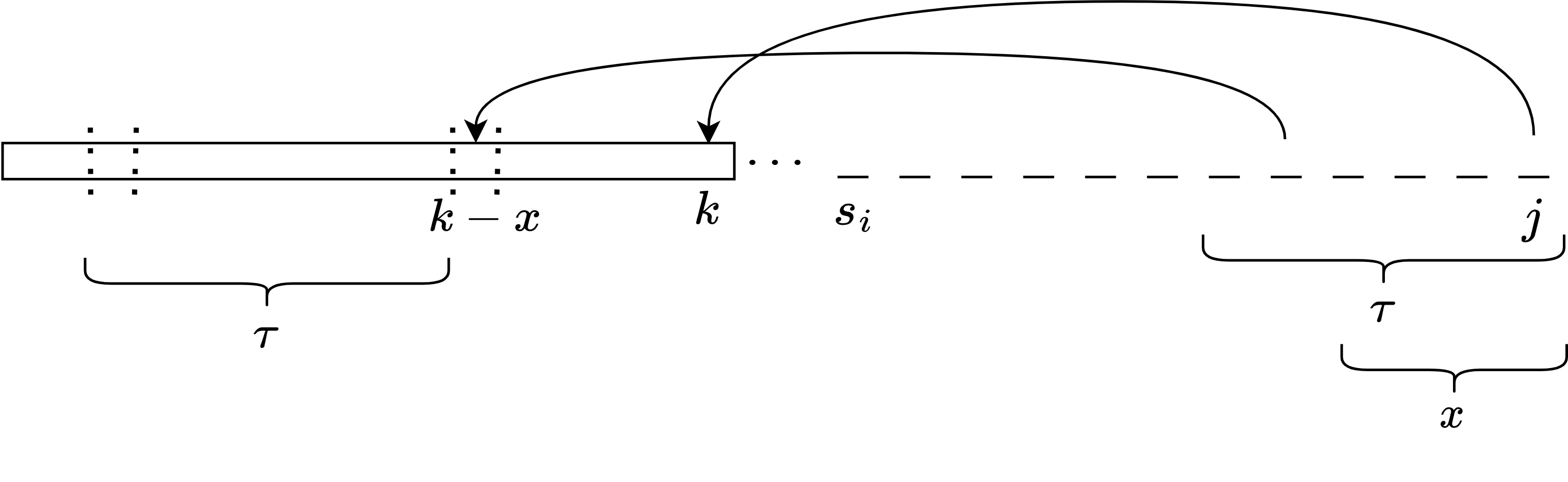}
    \end{minipage}
    
    \caption{The two cases given the proof of Lemma \ref{lem:tau_far_property}. On top is Case 1 where $T[s_i.. j]$ is not a potential factor. On bottom is Case 2 where $T[s_i.. j]$ is a potential factor. Here we are implying $k-x-1 \equiv 0 \mod \tau$.}
    \label{fig:monotonicity_proof}
\end{figure}

\begin{lemma}[Monotonicity of $\tau$-far property]
\label{lem:tau_far_property}
When finding a new factor starting at position $s_i$, if the $\tau$-far property holds for $j > s_i$, then it holds for $j-1$.
\end{lemma}

\begin{proof}
There are two cases; see Figure \ref{fig:monotonicity_proof}. Case 1: If $T[s_i.. j]$ is not a factor, since the $\tau$-far property holds for $j$, there exists an $h\in [j-\tau, j-1]$ and $k < s_i$ such that $T[s_i..h] = T[k-(h-s_i).. k]$ is a potential factor. Then this $h$ demonstrates that the $\tau$-far property holds for $j-1$. Case 2: Suppose instead that $T[s_i.. j]$ is a potential factor and matches some $T[k-(j-s_i).. k]$ where $k$ is the last position in a previous factor or $k < s_i$ and $k-1 \equiv 0 \mod \tau$. If $j-s_i+1 > \tau$, then there exists some $x \in [1, \tau]$ such that $k-(j-s_i) \leq k-x < k$ and $k-x-1 \equiv 0 \mod \tau$, hence $T[s_i.. j-x] = T[k-(j-s_i).. k-x]$ making $T[s_i.. j-x]$ a potential factor. Since $j-1 - \tau \leq j-x \leq j-1$, the $\tau$-far property holds for $j-1$. If instead $j-s_i+1 \leq \tau$, then $j-\tau \leq s_i \leq j$ and $T[s_i..s_i]$ is always a potential factor since either it is the first occurrence of a symbol or we can refer to the factor created by the first occurrence of $T[s_i]$. This proves that the 
property still holds for $j-1$.
\end{proof}

By Lemma \ref{lem:tau_far_property}, monotonicity holds for the $\tau$-far property when trying to find the next factor starting at position $s_i$ and to find the largest $j$ such that the $\tau$-far property holds, we can now use exponential search. 
At its core, we need to determine whether the $\tau$-far property holds for a given $j > s_i$. Once this largest $j$ is determined, the largest $h \in [\max(s_i, j-\tau),j]$ such that $T[s_i, h]$ is a potential factor must be determined as well.

We show a progression of algorithms to accomplish the above task. 
Firstly, we make some straight-forward, yet crucial, observations.
Let $S$ be any string. 
Since $\mathcal{P}_{i-1}$ is 
colexicographically sorted, all prefixes that have the same string (say $S$) as a suffix can be represented as a range of indices. 
 This range is empty when $S$ is not a suffix of any prefix in $\mathcal{P}_{i-1}$.
Moreover, this range can always be computed in $\tilde{O}(|S|)$ time using binary search.  However, if $S$ has an occurrence within the  $T[1..s_{i}-1]$ (i.e., the prefix seen thus far) and is specified by the start and end position of that occurrence, we can use LCE queries and improve the time for finding the range  to $\tilde{O}(1)$.



\paragraph{Next factor in $\tilde{O}(\tau + \ell_i)$ time:}
\label{sec:linear_sol_LCE}
For $h \in [\max(s_i, j-\tau), j]$, let $k_h \in [0, h-s_i+1]$ be the largest value such that $T[h-k_h+1.. h]$ is a suffix of a prefix in $\mathcal{P}_{i-1}$. 
We initialize $h = j$.
Since the prefixes in $\mathcal{P}_{i-1}$ are co-lexicographically sorted, we can find  $k_h$ in $\tilde{O}(k_h)$ time by using binary search on $\mathcal{P}_{i-1}$.
To do so, symbols are prepended one-by-one
and binary search is used to check if the corresponding sorted index range of $\mathcal{P}_{i-1}$ is non-empty.

Next we compute $k_h$ for $h = j -1, \dots, \max(s_i, j-\tau)$ in the descending order $h$. We keep track of 
$h' = \argmin_{y\in[h+1,j]} (y-k_y)$.
If $h'-k_{h'} +1 \leq h$, then $T[h'-k_{h'} + 1 .. h]$ has an occurrence in $T[1..s_i-1]$, and we can now use LCE queries 
to determine the range of $T[h'-k_{h'} + 1 .. h]$ in $\mathcal{P}_{i-1}$. If this range is empty, we conclude that 
$k_h < k_{h'}- (h'-h)$ and LCE can be used to find $k_h$ in $\tilde{O}(1)$ time. Otherwise, we proceed by prepending symbols one-by-one until $k_h$ is found.


The time per $h \in [\max(s_i, j-\tau), j]$ is $\tilde{O}(1)$ for LCE queries, in addition to $\tilde{O}(x_h)$ where $x_h$ is the number of symbols we prepended for $h$. Since we always use the smallest $h'-k_{h'}$ value seen thus far, $\sum_{h = j-\tau}^j x_h \leq j-s_i = O(\ell_i)$.
This makes it so checking if the $\tau$-far property holds for $j$ takes $\tilde{O}(\ell_i)$ time. The algorithm also identifies the rightmost $h \in [\max(s_i, j-\tau), j]$ such that $T[s_i, h]$ is a potential factor (if one exists). This only provides at best a near-linear time algorithm.

\paragraph{Next factor in $\tilde{O}(\sqrt{\tau\ell_i})$ time:}

Instead of prepending characters individually and using binary search after exhausting the reach of the LCE queries, we can instead apply the rightmost mismatch algorithm and use binary search on $\mathcal{P}_{i-1}$. Specifically, suppose that for a given $h \in [\max(s_i,j-\tau), j]$ we apply the LCE query and identify a non-empty range of prefixes in $\mathcal{P}_{i-1}$ with $T[w..h]$ as a suffix. On this set of prefixes we continue the search from $w-1$ downward using exponential search and identifying whether a mismatch occurs with the right-most mismatch algorithm. 

For $h \in [j-\tau, j]$ let $x_h$ now be the number of characters searched using exponential search and the right-most mismatch algorithm. As before $\sum_{h=j-\tau}^j x_h = \tilde{O}(j-s_i)$. The total time required for this is logarithmic factors from $\sum_{h=j-\tau}^j \sqrt{x_j} = \tilde{O}(\sqrt{\tau \ell_i})$. This will give us a sub-linear time algorithm if we choose $\tau$ appropriately; however, it will not be sufficient to obtain our goal.

\paragraph{Next factor in $\tilde{O}(\tau + \sqrt{\ell_i})$ time:}
Here we do not apply the rightmost-mismatch algorithm for every $h \in [\max(s_i,j-\tau), j]$.
Instead, for each $h$ we identify a set of prefixes in $\mathcal{P}_{i-1}$ such that $T[s_i.. h]$ shares a suffix of length at least $d_h = h-\max(s_i, j-\tau)+1$. This set is represented by the range of indices, $[s_h, e_h]$, in the sorted $\mathcal{P}_{i-1}$ corresponding to prefixes sharing this suffix of length $d_h = h-\max(s_i, j-\tau)+1$. By using the same LCE technique and prepending and stopping at index $\max(s_i, j-\tau)$, this can be accomplished in $\tilde{O}(\tau)$ time. After this, we have a set of ranges in $\mathcal{P}_{i-1}$. Note that for a given $h$, 
if we delete the last $d_h$ characters in each
prefix represented in $[s_h, e_h]$, they remain co-lexicographically sorted. 
We want to search for $T[s_i, j-\tau-1]$ as a suffix on these ranges, each with their appropriate suffix removed. Since rightmost-mismatch algorithm is costly, we can first union these ranges (each with their appropriate suffix removed), then use binary search.
However, merging these sorted ranges 
would be too costly. Instead, we can take advantage of the following lemma to avoid this cost.


\begin{lemma}[\cite{shiwangshiwang}]
\label{lem:k_stat}
Given $\tau$ sorted arrays $A_1$, ..., $A_\tau$ of $n$ elements in total,
the $x^{th}$ largest element in array formed by merging them can be found using $O(\tau \log \tau \cdot \log (n/\tau))$ comparisons.
\end{lemma}

Using the LCE data structure to compare any to prefixes, the $x^{th}$ largest element in the merged array can be found in $\tilde{O}(\tau)$ time. 
Using Lemma \ref{lem:k_stat}, we can find whichever rank prefix in the subset of $\mathcal{P}_{i-1}$ we are concerned with, then apply the rightmost mismatch algorithm to compare it to $T[s_i..j-\tau-1]$. Doing so, the total time needed for obtaining the next factor is $\tilde{O}(\tau + \sqrt{\ell_i})$.

\subsubsection{Determining the Number of Repetitions}

Similar to Section \ref{sec:linear_time}, we will repeat each use the rightmost mismatch algorithm $t$ times and take the majority solution if it exists, or one of the most frequent otherwise. The probability of the entire algorithm being incorrect is upper bounded by the sum of the probabilities of one of these quantum algorithm subroutines being incorrect.
Assuming the quantum subroutines return the correct solution with probability at least $\frac{2}{3}$, the probability of a majority solution not existing or being incorrect after $t$ repetitions is bound by $e^{-\frac{1}{18}t}$~\cite{valiant2019cs265}. 
For each evaluation of an index $j$ in the exponential search we use at most $\log n$ calls to the right-most mismatch algorithm (each comparing some prefix in the sorted set).
In the exponential search, at most $2\log n$ $j$ values, or equivalently, ranges of length $\tau$, are searched. Hence, per factor we use $2\log^2 n$ calls to  rightmost mismatch algorithm. There are at most $n$ factors, so the sum of these probabilities is bound by $2n \log^2 n \cdot e^{-\frac{1}{18}t}$. Setting less or equal to $\frac{1}{3}$ and solving for $t$ we see that it suffices $t = \lceil 18 \ln(6n \log^2 n) \rceil =O(\log n)$.

\subsubsection{Time and Query Complexity}

Taken over the entire string, the time complexity of finding the factors and updating the sorted order of the newly added prefixes is up to logarithmic factors bound by
\[
  \frac{n}{\tau} + z_{e + \tau} + \sum_{i=1}^{z_{e+\tau}}(\tau + \sqrt{\ell_i})
\leq 
 \frac{n}{\tau} + z_{e+\tau}+ \tau z_{e+\tau} + \sqrt{z_{e+\tau} n}  = O(\frac{n}{\tau}+\tau z_{e+\tau} + \sqrt{z_{e+\tau}n} )
\]
where the inequality follows from $\sum \ell_i = n$.  Combined with Lemma \ref{lem:LZ-End-tau}, which bounds $z_{e+\tau}$ to be logarithmic factors from $z$, and a logarithmic number of repetitions of each call to Grover's or rightmost mismatch algorithm, the total time complexity is $\tilde{O}(\sqrt{z n} + \tau z + \frac{n}{\tau})$. To minimize the time complexity we should set $\tau = \sqrt{n/z}$, bringing the total time to the desired $\tilde{O}(\sqrt{z n})$.



Note that we do not know $z$ in advance  to set $\tau$. However, the desired time complexity can be obtained by increasing our guess of $z$ as follows: 
Let $z_{guess}$ be initially $1$ and make $\tau_{guess} = \lceil \sqrt{n/z_{guess}}\rceil$ and run the above algorithm until either the entire factorization of the string $T$ is obtained or the number of factors encountered is greater than $z_{guess}$. 
For a given $z_{guess}$, the time complexity is bound by $\tilde{O}(\sqrt{z_{guess} n} + \tau_{guess}z_{guess} +  \frac{n}{\tau_{guess}})$, which is $\tilde{O}(\sqrt{z_{guess} n})$. If a complete factorization of $T$ is not obtained, we make $z_{guess} \gets 2 \cdot z_{guess}$, similarly update $\tau_{guess}$, and repeat our algorithm for the new $\tau_{guess}$. The total time taken over all guesses is logarithmic factors from $\sqrt{n}\sum_{i=1}^{\lceil \log z \rceil} (2^i)^{\frac{1}{2}}$ which, again, is $\tilde{O}(\sqrt{zn})$. 

The following lemma summarizes our result on LZ-End+$\tau$ factorization.
\begin{lemma}
Given a text $T$ of length $n$ having $z$ LZ77 factors, there exists a quantum algorithm that with probability at least $\frac{2}{3}$ obtains the LZ-End+$\tau$ factorization of $T$ in $\tilde{O}(\sqrt{zn})$ time and input queries.
\end{lemma}




\subsection{Obtaining the LZ77, SLP, RL-BWT Encodings}
\label{sec:other_encodings}

To obtain the other compressed encodings, we utilize the following result by Kempa and Kociumaka, stated here as Lemma \ref{lem:ipm_query}. We need the following definitions: a factor $F_i = (s_i, \ell_i)$ is called \emph{previous factor} $T[s_i.. s_i + \ell_i-1] = T[j.. j+\ell_i-1]$ for some $j < s_i$. We say a factorization $T = F_1, \hdots, F_f$ of a string is \emph{LZ77-like} if  each factor $F_i$ is non-empty and  $|F_i| > 1$ implies $F_i$ is a previous factor. Note that LZ-End+$\tau$ is LZ77-like with $f = O(z\log^2 n)$ as shown in Lemma \ref{lem:LZ-End-tau}.

\begin{lemma}[\cite{DBLP:journals/cacm/KempaK22} Thm.~6.11] 
\label{lem:ipm_query}
Given an LZ77-like factorization of a string $T[1,n]$ into $f$ factors, we can in $O(f \log^4 n)$ time construct a data structure that, for any pattern $P$ represented by its arbitrary occurrence in $T$, returns the leftmost occurrence of $P$ in $T$ in $O(\log^3 n)$ time.
\end{lemma}

Starting with the LZ-End+$\tau$ factorization obtained in Section \ref{sec:sublinear_time_alg}, we construct the data structure from Lemma \ref{lem:ipm_query}. To obtain the LZ77 factorization, we again work from left to right and apply exponential search to obtain the next factor. In particular, if the start of our $i^{th}$ factor is $s_i$ and $T[s_i..s_i+\ell_i-1]$ if the leftmost occurrence of the substring is at position $j < s_i$, then we continue the search by increasing $\ell$. Since $O(\log^3 n)$ time is used per query, we get that $O(\log^4 n)$ time is used to obtain each new factor. Therefore, once the data structure from Lemma \ref{lem:ipm_query} is constructed, the required time to obtain the LZ77 factorization is $O(z\log^4 n)$. The total time complexity of constructing all LZ77 factorization starting from the oracle for $T$ is 
$
\tilde{O}(\sqrt{z_{e+\tau} n} + z)
$, which is $\tilde{O}(\sqrt{z n})$, as summarized below. 

\begin{theorem}
\label{thm:alg_lz77}
Given a text $T$ of length $n$ having $z$ LZ77 factors, there exists a quantum algorithm that with probability at least $\frac{2}{3}$ obtains its LZ77 factorization in $\tilde{O}(\sqrt{zn})$ time and input queries.
\end{theorem}

To obtain the RL-BWT of the text we directly apply an algorithm by Kempa and Kociumaka.
In particular, they provide a Las-Vegas randomized algorithm that, given the LZ77 factorization of a text $T$ of length $n$, computes its RL-BWT in $O(z \log^8 n)$ time~(see Thm.~5.35 in \cite{DBLP:journals/cacm/KempaK22}).
Combined with $r = O(z\log^2 n)$~\cite{DBLP:journals/cacm/KempaK22}, we obtain the following result.

\begin{theorem}
\label{thm:alg_rlbwt}
Given a text $T$ of length $n$ with $r$ being the number of runs its BWT, there exists a quantum algorithm that with probability at least $\frac{2}{3}$ obtains its run-length encoded-BWT in $\tilde{O}(\sqrt{rn})$ time and input queries.
\end{theorem}

Obtaining a balanced CFG of size $\tilde{O}(z)$ is similarly an application of previous results, and can be obtained by applying either the original LZ77 to balanced grammar conversion algorithm of Rytter~\cite{DBLP:journals/tcs/Rytter03}, or more recent results for converting LZ77 encodings to grammars~\cite{kempa2021fast}, and even balanced straight-line programs~\cite{DBLP:journals/cacm/KempaK22}. 

\section{Obtaining the Suffix Array Index}
\label{sec:SA_index}

We start this section having obtained the RL-BWT and the LCE data structure for the input text. There are two main stages to the remaining algorithm for obtaining a fully functional suffix tree in compressed space. The first, is to obtain a less efficient index, which allows us to query the suffix array in $\tilde{O}(\tau)$ time per query. We accomplish this by applying a form of prefix doubling and alphabet replacement. These techniques allow us to `shortcut' the LF-mapping described in Section \ref{sec:prelimin}. Using this shortcutted LF-mapping, we then sample the suffix array values every $\tau$ text indices apart, similar to the construction of the original FM-index. Once this less efficient index construction is complete, we move to build the fully functional index designed by Gagie et al.~\cite{DBLP:journals/jacm/GagieNP20}. Doing so requires analyzing the values that must be found and stored, and showing that they can be computed in sufficiently few queries.

\subsection{Computing LF$^\tau$ and Suffix Array Samples}
\label{sec:lf_tau}

Recall that the $\LF$-mapping of an index $i$ of the BWT is defined as $\LF[i] = \ISA[\SA[i] - 1]$. The RL-BWT can be equipped rank-and-select structures in $\tilde{O}(r)$ time to support computation of the $\LF$-mapping of a given index in $O(1)$ time.

For a given BWT run corresponding to the interval $[s,e]$, we have for $i \in [s,e-1]$ that $\LF[i+1] - \LF[i] = 1$, i.e., intervals contained in BWT runs are mapped on to intervals by the LF-mapping. If we applied the LF-mapping again to each $i \in [\LF[s], \LF[e]]$, the BWT-runs occurring $[\LF[s], \LF[e]]$ may split the $[\LF[s], \LF[e]]$ interval. We define the \emph{pull-back} of a mapping $\LF^2$ as $[s,e] \mapsto [s_1, e_1], [s_2, e_2],\hdots, [s_k, e_k]$ that satisfies $s_1 = s$, $s_{i+1} = e_i + 1$ for $i \in [1, k-1]$, $e_k = e$, $\BWT[\LF[s_{i+1}]] \neq \BWT[\LF[e_{i}]]$, and $j \in [s_i, e_i-1]$, $i\in[1,k]$ implies $\BWT[\LF[j]] = \BWT[\LF[j+1]]$. We also assign to each interval $[s_i, e_i]$ created by the $\LF^2$ pull-back: (i) a string of length two specifically, $[s_i, e_i]$ is assigned the string $\BWT[\LF[s_i]] \circ \BWT[s_i]$ (ii) the indices that $s_i$ and $e_i$ map to, that is $\LF^2[s_i]$, $\LF^2[e_1]$. See Figure \ref{fig:lf_tau}.

\begin{figure}
\centering
\includegraphics[width=.7\textwidth]{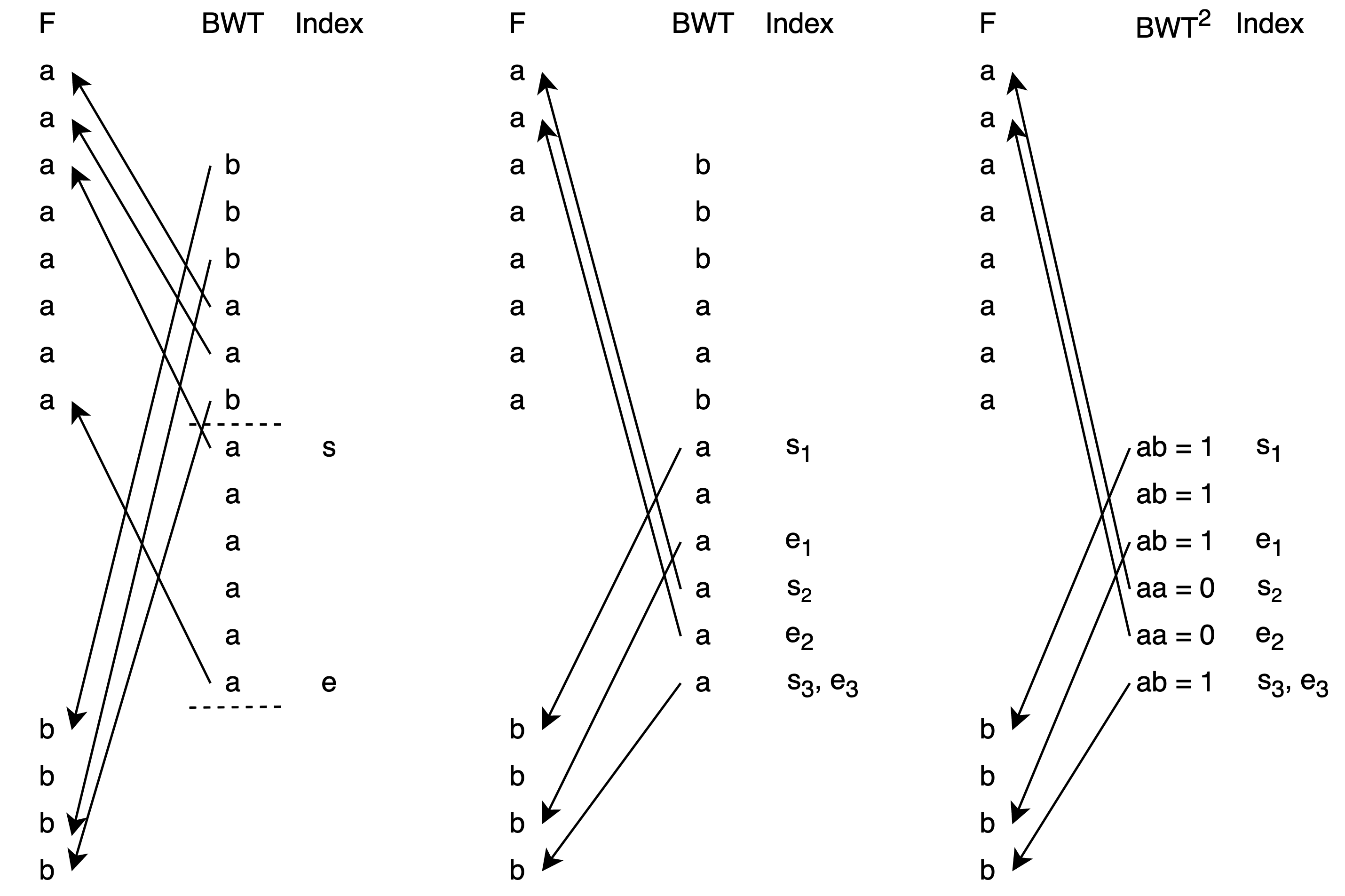}
\caption{
(Left) The initial LF-mapping from the interval $[s,e]$. (Middle) The intervals created for $[s,e]$ by the $\LF^2$ pull-back. (Right) Alphabet replacement is applied to each interval to create $\BWT^2$. Note that the contents of the F-column are not important and are not used. 
}
\label{fig:lf_tau}
\end{figure}

Observe that when applying the LF-mapping to two distinct BWT-runs, they must map onto disjoint intervals. Hence, each BWT run boundary appears in exactly one interval $[\LF[s], \LF[e]]$ where $[s,e]$ is a BWT run. As a consequence, if we apply $\LF^2$ pull-back to all BWT run intervals, and split each BWT run interval according to its $\LF^2$ pull-back, the number of intervals at most doubles.   

For a given $\tau$ that is a power of two, we next describe how to apply this pull-back technique and alphabet replacement to precompute mappings for each BWT run. These precomputed mappings make it so that given any $i \in [1,n]$, we can compute $\LF^\tau[i]$ in $\tilde{O}(1)$ time. The time and space needed to precompute these mappings are $\tilde{O}(\tau r)$. 

We start as above and compute the $\LF^2$ pull-back for every BWT-run interval. We then replace each distinct string of length two with a new symbol. This could be accomplished, for example, by sorting and replacing each string by its rank. However, it should be noted that the order of these new symbols is not important. Doing this assigns each index in $[1,n]$ a new symbol. We denote this assignment as $\BWT^2$ and observe that the run-length encoded $\BWT^2$ is found by iterating through each $\LF^2$ pull-back. 

We now repeat this entire process for the runs in $\BWT^2[j]$. Doing so gives us for each interval $[s,e]$ corresponding to a $\BWT^2$ run a set of intervals $[s,e] \mapsto [s_1, e_1], [s_2, e_2],\hdots, [s_k, e_k]$ that satisfies $s_1 = s$, $s_{i+1} = e_i + 1$ for $i \in [1, k-1]$, $e_k = e$, $\BWT[\LF^2[s_{i+1}]] \neq \BWT[\LF^2[e_{i}]]$, and $j \in [s_i, e_i-1]$, $i\in[1,k]$ implies $\BWT[\LF^2[j]] = \BWT[\LF^2[j+1]]$. We call this the $\LF^4$ pull-back. Next we again apply alphabet replacement, defining $\BWT^4$ accordingly. The next iteration will compute the $\LF^8$ pull-back and $\BWT^8$. Repeating $\log \tau$ times, we get a set of intervals corresponding to $\LF^\tau$ pull-back. 

Recall that each pull-back step at most doubles the number of intervals. Hence, by continuing this process $\log \tau$ times, the number of intervals created by corresponding $\LF^\tau$ pull-back is $2^{\log \tau} r = \tau r$. For a given index $i$, to compute $\LF^\tau[i]$ we look at the $\LF^\tau$ pull-back. Suppose that $i \in [s', e']$ where $[s', e']$ is interval computed in the $\LF^\tau$ pull-back. We look at  the mapped onto interval $[\LF^\tau[s'], \LF^\tau[e']]$ which we have stored as well, and take $\LF^\tau[i] = \LF^\tau[s'] + (i-s')$.

We are now ready to obtain our suffix array samples. 
We start with the position for the lexicographically smallest suffix (which, by concatenating a special symbol $\$$ to $T$, we can assume is the rightmost suffix).
Make $\tau$  equal to the smallest power of two greater or equal to $\lceil (\frac{n}{r})^\frac{1}{2} \rceil$.
Utilizing LF$^\tau$ we compute and store $\frac{n}{\tau}$ suffix array values even spaced by text position and their corresponding positions in the RL-BWT. This makes the $\SA$ value of any position in the BWT obtainable in $\tau$ applications of $\LF^1$ and computatible in $\tilde{O}(\tau)$ time.  Inverse suffix array, $\ISA$, queries can be computed with additional logarithmic factor overhead by using the (standard) LCE data structure. For a give range of indexes in the suffix array $[s,e]$, the longest common prefix, $\LCP$, of all suffixes $T[\SA[i]..n]$, $i \in [s,e]$ in $\tilde{O}(\tau)$. This can be computed by first finding $\SA[s]$ and $\SA[e]$ and then using the (standard) LCE data structure to find the length of the shared prefix of $T[\SA[s]..n]$ and $T[\SA[e]..n]$.

\subsection{Constructing a Suffix Array Index}
The r-index~\cite{DBLP:conf/soda/GagieNP18} for locating and counting pattern occurrences can be constructed by sampling suffix array values at the boundaries of BWT runs. Requiring $\tilde{O}(r)$ queries, this takes $\tilde{O}(r\tau)$ time, and with $\tau = \lceil \left(\frac{n}{r}\right)^\frac{1}{2} \rceil$, has the desired time complexity.
Expanding on this, we will show that one can construct the more functional $\SA$ index by Gagie et al.~\cite{DBLP:journals/jacm/GagieNP20} using $\tilde{O}(r)$ $\SA$ queries, yielding the desired time complexity. We omit the $\tilde{O}(1)$ query process used on that data structure, as its not relevant for the work here. Before describing the index, we describe the differential suffix array, $\DSA$, where $\DSA[i] = \SA[i] - \SA[i-1]$ for all $i > 1$. The key property is that the LF-mapping applied to a portion of the $\DSA$ completely contained within a BWT run preserves the $\DSA$ values.

\begin{lemma}[~\cite{DBLP:journals/jacm/GagieNP20}]
Let $[i-1..i+s]$ be within a BWT run, for some $1 < i \leq n$ and $0 \leq s \leq n-i$. Then there exists $q \neq i$ such that $\DSA[q..q+s] = \DSA[i..i+s]$ and $[q-1..q+s]$ contains the first position of a BWT run.
\end{lemma}

For an arbitrary range $[s..e]$ we can obtain such a $q$ in $\tilde{O}(\tau)$ time using the previously computed values from Section \ref{sec:lf_tau}. 

\begin{lemma}
\label{lem:LF_k}
Given that $\SA[i]$ for arbitrary $i$ can be computed in $\tilde{O}(\tau)$ time and (reversed) LCE queries in $\tilde{O}(1)$ time, for an given range $[s..e]$ we can find $k$, $s'$, $e'$ such that $k \geq 0$ is the smallest value where $\LF^k([s,e]) = [s', e']$ and $[s',e']$ contains the start of BWT-run.
\end{lemma}

\begin{proof}


We assume $k \geq 1$ other we determine from the RL-BWT that $[s,e]$ contains a run. 



We use exponential search on $k$. For a given $k$, we compute $s' = \ISA[\SA[s]-k]$ and $e' = \ISA[\SA[e]-k]$. We check whether $e'-s' =e -s$ and whether $\LCP([s', e']) \geq k$. If both of these conditions hold then no run-boundary has been encountered yet. The largest largest $k$, $s'$, and $e'$ for which these conditions hold is returned. Utilizing the results from Section \ref{sec:lf_tau}, this takes $\tilde{O}(\tau)$ time per $\SA$ or $\ISA$ query, resulting in $\tilde{O}(\tau)$ being needed overall.
\end{proof}

We next describe the $\SA$ index construction from \cite{DBLP:journals/jacm/GagieNP20}.
The index consists of $O(\log \frac{n}{r})$ levels. The top level ($l = 0$) consists of $r$ ranges, or 'blocks', corresponding to indices $[1..\frac{n}{r}]$, $[\frac{n}{r} + 1..2\frac{n}{r}]$, ..., $[n-\frac{n}{r}..n]$. For $l > 0$, define $s_l = \frac{n}{r2^{l-1}}$. For every position $q$ that starts a run in the BWT we have blocks $X_{l,q}^1$ corresponding to $[q-s_l+1..q]$, $X_{l,q}^2$ corresponding to $[q+1..q+s_l]$. Each of these is further subdivided into two half-blocks $X_{l,q}^k = X_{l,q}^k[1..\frac{s_l}{2}]X_{l,q}^k[\frac{s_l}{2}+1..s_l]$, $k \in \{1,2\}$. There are also blocks $[q-s_l + \frac{s_l}{4}+1..q-\frac{s_l}{4}]$, $[q-\frac{s_l}{4}+1..q+\frac{s_l}{4}]$ and $[q+\frac{s_l}{4}+1..q+s_l-\frac{s_l}{4}]$.

For any given block $X$ at level $l=0$ or half-block $X$ at level $l > 0$ where $l$ is not the last level, we need to store a pointer to a block on level $l+1$ containing $q^*$ where $q^* \in \{q, q+1\}$ and $q$ is the start of a BWT run and there exists a matching instance of the $\DSA$ for the range of $X$ containing $q^*$. We find this by applying Lemma \ref{lem:LF_k} above.

This pointer also stores values $\langle q^*, \mathit{off}, \Delta \rangle$. The value $\mathit{off}$ is define so that $q^*-s_{l+1} + \mathit{off}$ gives the starting position of this instance. This can be obtained knowing the range $[s',e']$ described above, and $q$ from the RL-BWT. That is, $\mathit{off} = s' - q^* + s_{l+1}$. The value $\Delta = \SA[q-s_{\ell+1}] - \SA[q^*-s_{\ell+1} + \mathit{off} - 1]$ is obtainable using two $\SA$ queries. Every level-0 block $[q'+1..q'+s_\ell]$ also stores $\SA[q']$ and half-blocks $X' = [q'+1..q'+\frac{s_{l+1}}{2}]$ store $\SA[q'] - \SA[q-s_{l+1}]$.
The last level of the structure, level $l^*$ explicitly stores the all $\DSA$ values for $X_{l^*, q}^1$ and $X_{l^*, q}^2$ for $q$ starting a run of the BWT. This uses $O(r s_{l^*}) = O(r \log \frac{n}{r})$ space.

Over the creation of all pointers, $\DSA$ values needed, and $\SA$ values needed, a total of $\tilde{O}(r)$ queries to the $\SA$ array are needed. The result is that we can construct the above structure using $\tilde{O}(r \tau)$, or equivalently$\tilde{O}(\sqrt{rn})$, time. This establishes the following theorem.



\begin{theorem}
Given a text $T$ of length $n$ there exists a quantum algorithm that with probability at least $\frac{2}{3}$ obtains a fully-functional index for the text in  $\tilde{O}(\sqrt{zn})$ time. Once constructed requires $\tilde{O}(z)$ space and supports suffix array, inverse suffix array, and longest common extension queries all in $\tilde{O}(1)$ time.
\end{theorem}

\section{Applications}
\label{sec:applications_}
Having constructed the fully functional index for $T$, we can now solve a number of other problems in sub-linear time:

\paragraph{Longest Common Substring:} Given two strings $S_1$ and $S_2$, we let $z_{1,2}$ be number of LZ77 factors of $S_1\$S_2$. In $\tilde{O}(\sqrt{z_{1,2}(|S_1| + |S_2|)})$ time we construct the LZ77 parse and LCE-data structure for the $S_1\$S_2$. 

\begin{lemma}
\label{lem:lcs}
If $T$ is the longest common substring of $S_1$ and $S_2$, then there exists $S_1[i_1..j_1] = T$ and $S_2[i_2..j_2] = T$ such that in the $\BWT$ and $\ISA$ for $S_1\$S_2$ we have $|\ISA[i_1] - \ISA[|S_1|+1+i_2]| = 1$. Moreover, for this instance, $\BWT[\ISA[i_1]] \neq \BWT[\ISA[i_2]]$.
\end{lemma}

\begin{proof}
Suppose all instances of $T$ that $|\ISA[i_1] - \ISA[|S_1|+1+i_2]| > 1$ and let $i_1$ and $i_2$ be such that $|\ISA[i_1] - \ISA[|S_1|+1+i_2]|$ is minimized. Suppose WLOG that $\ISA[i_1] < \ISA[|S_1|+1+i_2]$. Then there must exists indices $x$ and $y$ such that $\ISA[i_1] \leq x \leq y \leq \ISA[|S_1|+1+i_2]$ such that $y-x = 1$, $\SA[x]$ $\SA[y]$ are indices into ranges for different strings, i.e., $\SA[x] \in [1, |S_1|]$ and $\SA[y] \in [|S_1|+2, |S_1|+|S_2|+1]$ or $\SA[x] \in [|S_1|+2, |S_1|+|S_2| + 1]$ and $\SA[y] \in [1, |S_1|]$, and the $\LCE(\SA[x], \SA[y]) \geq \LCE(\SA[i_1], \SA[i_2])$. Hence, the prefix $T$ is shared by $S_1\$S_2[x..n]$ and $S_1\$S_2[y..n]$, a contradiction.  At the same time if $\BWT[i_1] = \BWT[|S_1|+1+i_2]$, then the longest common substring could be extended to the left, a contradiction.
\end{proof}

Based on Lemma \ref{lem:lcs}, we can find the longest common string of $S_1$ and $S_2$ by checking the runs in BWT of $S_1\$S_2$, checking if adjacent $\SA$ values correspond to suffixes of different strings, and through the LCE queries, taking the pair with the largest shared prefix.

Note that $\Omega(\sqrt{|S_1| + |S_2|})$ time is necessary for the problem when $z = \Theta(1)$, as, by Lemma \ref{lem:lb_for_check_if_0} determining whether $S_1 = 0^{n}$ or whether $S_1$ contains a single $1$ requires $\Omega(\sqrt{n})$ time. The reduction simply sets $S_2 = `1\textrm'$.


\paragraph{Maximal Unique Matches:} Given two strings $S_1$ and $S_2$, we can identify all maximal unique matches in $\tilde{O}(r)$ additional time after constructing the RL-BWT and our index for $S_1\$S_2$. To do so we iterate through all run boundaries in the RL-BWT. We wish to identify all occurrences where: 
\begin{itemize}
\item $\BWT[i] \neq \BWT[i+1]$; 
\item $\SA[i] \in [1,|S_1|]$ and $\SA[i+1] \in [|S_1|+2, |S_1|+|S_2|]$, or $\SA[i] \in [|S_1|+2, |S_1|+|S_2|]$ and $\SA[i+1] \in [1, |S_1|]$; 
\item and either 

\begin{itemize}
    \item $i = 1$ and $\LCE(\SA[i], \SA[i+1]) > \LCE(\SA[i+1], \SA[i+2])$, or
    
    \item $i > 1$ and $i+1 < n$ and \[
    \LCE(\SA[i], \SA[i+1]) > \LCE(\SA[i-1], \SA[i]), \LCE(\SA[i+1], \SA[i+2])\] or
    
    \item $i > 1$ and $i + 1 = n$ and $\LCE(\SA[i], \SA[i+1]) > \LCE(\SA[i-1], \SA[i])$.
\end{itemize}
\end{itemize}


\paragraph{Lyndon Factorization:} 
The Lyndon factors of a string are determined by the values in $\ISA$ that are smaller than any previous value, i.e., $i \in [1,n]$ s.t. $\ISA[i] < \ISA[j]$ for all $j < i$. After constructing the $\SA$ index above in $\tilde{O}(\sqrt{zn})$, each $\ISA$ value can be queried in $\tilde{O}(1)$ time as well. 
Let $f$ denote the total number of Lyndon factors. Letting $i_j$ by the index where the $j^{th}$ Lyndon factor.
To find all Lyndon factor, we proceed from left to right keeping track of the minimum $\ISA$ value encountered thus far. Initially, this is $\ISA[i_1] = \ISA[1]$. We use exponential search on the right most boundary of the subarray being search and Grover's search to identify the left-most index  $x \in [i_j+1,n]$ such that $\ISA[x] < \ISA[i_j]$. If one is encountered we set $i_{j+1} = x$ and continue. Thanks to the exponential search, the total time taken is logarithmic factors from
\[
\sum_{j=1}^f \sqrt{i_{j+1} - i_j} \leq \sqrt{fn}.
\]
At the same time, the number of Lyndon factors $f$ is always bound by $\tilde{O}(z)$~\cite{DBLP:conf/stacs/KarkkainenKNPS17,DBLP:conf/cpm/UrabeNIBT19}. This yields the desired time complexity of $\tilde{O}(\sqrt{zn})$.



\paragraph{Q-gram Frequencies:} 
We need to identify the nodes $v$ of the suffix tree at string depth $q$ and at the  number of leaves in thier respective subtrees. A string matching the root to $v$ path where $v$ has string depth $q$ is a q-gram and the size of the subtree its frequency.  To do this, we start with $i = 1$. Our goal is find the largest $j$ such that $\LCE(\SA[i], \SA[j]) \geq q$. This can be found using exponential search on $j$ starting with $j = i$. Once this largest $j$ is found, the size of the subtree is equal to $j-i+1$. The q-gram itself is $T[\SA[i]..\SA[i]+q-1]$. We then make $i = j+1$ and repeat this process until the $j$ found through exponential search equals $n$.  The overall time after index construction is $\tilde{O}(occ)$ where $occ$ is the number of q-grams. 

The $\tilde{O}(\sqrt{zn})$ time complexity is near optimal as a straight forward reduction from the Threshold problem discussed next with $q = 1$ and binary strings indicates that $\Omega(\sqrt{zn})$ input queries are required.



\vspace{1em}
These are likely only of a sampling of the problems that can be solved using the proposed techniques. It is also worth restating that the compressed forms of the text can be obtained in $O(\sqrt{zn})$ input queries, hence all string problems can be solved with that many input queries, albeit, perhaps with greater time needed. 

\section{Lower Bounds}
\label{sec:lb}

In this section, we provide reductions from the Threshold Problem, which is defined as:

\begin{problem}[Threshold Problem]
Given an oracle $f:\{1,\hdots, n\} \rightarrow \{0,1\}$ and integer $t \geq 0$, determine if there exists at least $t$ inputs $i$ such that $f(i) = 1$. i.e., if  $S = |\{i \in \{1,\hdots, n\} \mid f(i) = 1\}|$, is $|S| \geq t$?
\end{problem}

Known lower bounds on the quantum query complexity state that for $0\leq t < \frac{n}{2}$ at least $\Omega(\sqrt{(t+1)n})$ queries to the input oracle are required to solve the Threshold Problem~\cite{DBLP:journals/eatcs/HoyerS05,DBLP:conf/stoc/Paturi92}. One obstacle in using the Threshold problem to establish hardness results is that we cannot make assumptions concerning the size of the set $|S|$ and we do not assume knowledge of $z$ for the problem of finding the LZ77 factorization. As a result, additional assumptions on our model are used here. The assumption we work under is that the quantum algorithm can be halted when the number of input queries exceeds some predefined threshold based on $z$ and $n$.

The following Lemma helps to relate the size of the set $S$ and the number of LZ77 factors, $z$, in the binary string representation of $f$.

\begin{lemma}
\label{lem:z_leq_t}
If the string representation $T = \bigcirc_{i=1}^n f(i)$ has $|S| \geq 0$ ones, then the LZ77 factorization size of $T$ is $z \leq 3|S| + 2$.
\end{lemma}

\begin{proof}
Each run of $0$'s contributes at most two factors to the factorization and each $1$ symbol contributes at most one factor, hence each $0^x1^y$ substring, $x \geq 0$, $y \geq 1$ accounts for at most $3$ factors and the $+2$ accounts for a possible suffix of all $0$'s.
\end{proof}

For finding the RL-BWT we can also establish a similar result
\begin{lemma}
\label{lem:rlbwt_runs}
If the string representation $T = \bigcirc_{i=1}^n f(i)$ has $|S| \geq 0$ ones, then the number of runs in BWT of $T$ is $r \leq 2|S|+1$.
\end{lemma}

\begin{proof}
Consider any permutation of a binary string with $|S|$ 1's and $n-|S|$. To maximize the number of runs, we alternate between $1$'s and $0$'s. This creates at most two runs per every $1$, in addition to a leading run, hence at most $2|S| + 1$. Note that the BWT of $T$ is a permutation.
\end{proof}

\subsubsection{Basic Lower Bounds - Hardness for Extreme $z$ and $r$}

For the case where $2 \leq z \leq 5$ we can apply the adversarial method of Ambainis~\cite{DBLP:conf/stoc/Ambainis00}. 
In particular, one can show through an application of the adversarial method that $\Omega(\sqrt{n})$ queries are required to determine if $|S| \geq 1$ even given the promise that $T$ contains either no $1$'s or exactly one $1$. See Appendix \ref{sec:threshold} for the proof. In the case where $|S| = 0$, we have $z = 2$ (and $r = 1$), and in the case where $|S| \geq 1$, we have $3 \leq z \leq 5$ and $2 \leq r \leq 3$. Note also that even if we only obtained a suffix array index, and not the compressed encodings, this would also let us determine if there exists a $1$ in $T$ with a single additional query.

To show hardness for $z = \Theta(n)$, we consider the problem of determining the $\left(\frac{n}{2}\right)^{th}$ largest value output from an oracle $f$ that outputs distinct values for each index. This un-ordered searching problem has known $\tilde{O}(n)$ lower bounds~\cite{DBLP:conf/stoc/NayakW99}. For the string representation $T = \bigcirc_{i=1}^n f(i)$, if we could obtain an LZ77 encoding of $T$ using $o(z)$, or equivalently $o(n)$, input queries, then we can decompress the result and return the middle element to solve the un-ordered searching problem. Similarly, if we could obtain the RL-BWT of $T$ in $o(r)$ queries, we could decompress it to obtain the middle element in $o(n)$ queries. Note also that even if we only obtained a suffix array index, and not the compressed encoding, with a single additional query we can find the median element.

The above observations yield the following results that hold even for models where a quantum algorithm cannot be halted based on the current number of input queries. 

\begin{theorem}
\label{thm:lb_no_qram}
Obtaining the LZ77 factorization of a text $T[1..n]$ with $z$ LZ77 factors, or RL-BWT with $r$ runs, requires $\Omega(\sqrt{zn})$ queries $($$\Omega(\sqrt{rn})$ queries resp.$)$ when $z, r = \Theta(1)$ or $z,r = \Theta(n)$.
\end{theorem}

\begin{theorem}
Constructing a data structure that supports $o(\sqrt{zn})$ time 
 suffix array queries of a text $T[1..n]$ with $z$ LZ77 factors, or RL-BWT with $r$ runs, requires  $($$\Omega(\sqrt{zn})$ $($$\Omega(\sqrt{rn})$ resp.$)$ input queries when $z, r = \Theta(1)$ or $z,r = \Theta(n)$.
\end{theorem}

\subsection{Parameterized Hardness of Obtaining the LZ77 Factorization and RL-BWT}
\label{sec:hardness_parse}

We next show that under the stated assumption regarding the ability to halt the program, $\Omega(\sqrt{zn})$ input queries are required for obtaining the LZ77 factorization strings for a wide regime of different possible $z$ values.
Let the instance of the Threshold Problem with $f:\{1,\hdots, n\} \rightarrow \{0,1\}$ and $t$ be given as a function of $n$. Assume for the sake of contradiction that there exists an algorithm for finding the LZ77 factorization of $f$ in some known $q(n,z) = o(\sqrt{zn})$ queries for strings whenever $z \in [t^{1-\varepsilon}, t]$ for some constant $\varepsilon >0$. As an example, suppose that there exists an algorithm for finding the LZ77 factorization in $q(n,z)$ time whenever $z \in [n^{\frac{1}{2} - \varepsilon}, n^{\frac{1}{2}}]$ for some constant $\varepsilon$. 

We assume that we can halt the algorithm if the number of input queries exceeds a threshold $\kappa(n,t) \in \omega(q(n,t)) \cap o(\sqrt{(t+1)n})$. In particular, we can take $\kappa(n,t) = q(n,t) \cdot \left(\frac{\sqrt{(t+1)n}}{q(n,t)} \right)^\frac{1}{2}$, which has $\lim_{n\rightarrow \infty} \kappa(n,t) / \sqrt{(t+1)n} = 0$ and $\lim_{n\rightarrow \infty} q(n,t) / \kappa(n,t) = 0$. To solve the Threshold Problem with the LZ77 Factorization algorithm, do as follows:

\begin{enumerate}

\item Run the $O(\sqrt{zn})$ query algorithm from Section \ref{sec:oracle_id_algorithm}, but halt if the number of input queries reaches $\kappa(n,t)$. If we obtain a complete encoding without halting, we output the solution, otherwise, we continue to Step 2. This solves the Threshold Problem instance in the case where $z \leq t^{1-\varepsilon}$ in $o(\sqrt{(t+1)n})$ input queries.

\item We next run the assumed algorithm with query complexity $q(n,z)$, but again halt if the number of input queries exceeds $\kappa(n,t)$. 
If we halt with a completed encoding, we output the solution based on the complete encoding of $T$. If we do not halt with a completed encoding, we output that $|S| \geq t$.
\end{enumerate}

In the case where $|S| < t$, we have $z \leq 3|S| + 2 \leq 3t + 2$ and our algorithm solves the problem in either $O(\sqrt{zn})$ queries if $z \leq t^{1-\varepsilon}$, or otherwise in $q(n,z) = o(\sqrt{(t+1)n})$ queries if $z \in [t^{1-\varepsilon}, t]$. In the case where $|S| \geq t$, the number of queries is still bound by $\kappa(n,t) = o(\sqrt{(t+1)n})$. Regardless, we solve the Threshold Problem with $o(\sqrt{(t+1)n})$ input queries. As such, the assumption that such an algorithm exists contradicts the known lower bounds, proving Theorem \ref{thm:lb_}.

Using Lemma \ref{lem:rlbwt_runs}, a nearly identical argument holds for computing the RL-BWT of $T$ with $z$ replaced by $r$, showing that $\Omega(\sqrt{rn})$ queries are required. The above proves the following Lemma.

\begin{theorem}
\label{thm:lb_}
Under a model where a given quantum algorithm can be halted when the number of input queries exceeds some predefined threshold, obtaining the LZ77 factorization of text $T[1..n]$ with $z$ LZ77 factors, or BWT with $r$ runs, requires $\Omega(\sqrt{zn})$ queries $($$\Omega(\sqrt{rn})$ queries resp.$)$. In particular, no algorithm exists using $o(\sqrt{zn})$ queries $(o(\sqrt{rn})$ queries$)$ for all texts with $z \in [t^{1-\varepsilon}, t]$ $(r \in [t^{1-\varepsilon}, t])$ where $t$ can be any function of $n$ and $\varepsilon > 0$ is any constant. This holds for alphabets of size two or greater.
\end{theorem}

\subsection{Parameterized Lower Bounds for Computing the Value $z$}
\label{sec:hardness_z_value}

We next turn our attention to the problem of determining the number of factors, $z$, in the LZ77-factorization. This is potentially an easier problem than actually computing the factorization. We will again use the assumption that the algorithm can be halted once the number of input queries exceeds some specified threshold. Unlike the proof for the hardness of computing the actual LZ77 factorization, our proof uses a larger integer alphabet rather than a binary one.

Given inputs $f:\{1,\hdots, n\} \rightarrow \{0,1\}$ and $t(n) \geq 1$ to the Threshold Problem, we construct an input oracle for a string $T$ for the problem of determining the size of the LZ77 factorization.
We first define the function $s: \{0,1\}\times \{1,\hdots, n\} \rightarrow \{0, 1,\hdots, n\}$ as:
\[
s(f(i),i) = 
\begin{cases}
0 & f(i) = 0\\
i & f(i) = 1
\end{cases}
\]
Our reduction will create an oracle for a string $T= 0^{2n} \circ \$ \circ  \left(\bigcirc_{i=1}^n 0 \circ s(f(i),i) \right) \circ 0$.
Formally, we define the oracle $T:\{1,\hdots, 4n+2\} \rightarrow \{0, \$\} \cup \{1,\hdots, n\}$ where
\[
T[i] = \begin{cases}
0 & 1 \leq i \leq 2n \text{ or } i > 2n + 1 \text{ and } i \text{ is even}\\
\$ & i = 2n+1 \\
s\left(f(\frac{i - (2n + 1)}{2}), \frac{i - (2n + 1)}{2}\right) & i > 2n + 1 \text{ and } i \text{ is odd}
\end{cases}
\]
Note that every symbol in $T$ can be computed in constant time given access to $f$. The correctness of the reduction will follow from Lemma \ref{lem:thresh_eq_z}.

\begin{lemma}
\label{lem:thresh_eq_z}
The LZ77 factorization of $T$ constructed above has $z = 2|S| + 4$.
\end{lemma}

\begin{proof}
The prefix $0^{2n} \circ \$$ will always require exactly 3 factors, having the factorization $(`0')$, $(1,2n-1)$, $(`\$')$. Any subsequent run of $0$'s will only require one factor to encode. This is since any run of $0$'s following this prefix is of length at most $2n$. Next, note that the $0$ following $\$$ is the start of a new factor. Following this $0$ symbol, we claim that every $i$ where $f(i) = 1$ contributes exactly two new factors. This is due to a new factor being created for the leftmost occurrence of the symbol $i$ followed by a new factor for the run of $0$'s beginning after the symbol $i$. Observe that the first $0$ following the $\$$ is necessary for the lemma, otherwise, the cases where $f(1)= 1$ would result in a different number of factors.
\end{proof}

Similar to Section \ref{sec:hardness_parse}, assume for the sake of contradiction that we have an algorithm that solves the problem of determining the value $z$ in $q(n,z) = o(\sqrt{zn})$ for $z = [t^{1-\varepsilon}, t]$ and that we can terminate the algorithm if the number of input queries exceeds $\kappa(n,z) \in \omega(q(n,z)) \cap o(\sqrt{zn})$. To solve the instance of the Threshold Problem, the two-step algorithm from Section \ref{sec:hardness_parse} is applied to the oracle for $T$, resulting in a solution using $o(\sqrt{(t+1)n})$ input queries. This demonstrates the following theorem.

\begin{theorem}
\label{thm:lb_value}
Under a model where a given quantum algorithm can be halted when the number of input queries exceeds some predefined threshold, obtaining the number of factors in the LZ77 factorization of text $T[1..n]$ with $z$ LZ77 factors,  requires $\Omega(\sqrt{zn})$ queries. In particular, no algorithm exists using $o(\sqrt{zn})$ queries for all texts with $z \in [t^{1-\varepsilon}, t]$ where $t$ can be any function of $n$ and $\varepsilon > 0$ is any constant.
\end{theorem}

\section{Open Problems}
\label{sec:discussion}
A main implication of this work is that for strings that are compressible through substitution-based compression, many problems can be solved faster with quantum computing. We leave open several problems, including: Can we prove comparable lower bounds (of the form $\Omega(\sqrt{rn})$, $\Omega(\sqrt{\delta n})$, etc.) for algorithms that only compute the other compressibility measure values like $r$ and $\delta$, or even $z$ in the case of binary alphabets? Ideally, these would hold over the entire range of possible values for each measure. Also, can we prove lower bounds parameterized by $z$ and overall ranges of $z$ for the applications we provided, LCS, q-gram frequencies, etc.

\bibliography{ref.bib}

\begin{thebibliography}{10}

\bibitem{ablayev2020quantum}
Farid Ablayev, Marat Ablayev, Kamil Khadiev, Nailya Salihova, and Alexander
  Vasiliev.
\newblock Quantum algorithms for string processing.
\newblock {\em arXiv preprint arXiv:2012.00372}, 2020.

\bibitem{DBLP:conf/soda/AkmalJ22}
Shyan Akmal and Ce~Jin.
\newblock Near-optimal quantum algorithms for string problems.
\newblock In Joseph~(Seffi) Naor and Niv Buchbinder, editors, {\em Proceedings
  of the 2022 {ACM-SIAM} Symposium on Discrete Algorithms, {SODA} 2022, Virtual
  Conference / Alexandria, VA, USA, January 9 - 12, 2022}, pages 2791--2832.
  {SIAM}, 2022.
\newblock \href {https://doi.org/10.1137/1.9781611977073.109}
  {\path{doi:10.1137/1.9781611977073.109}}.

\bibitem{DBLP:conf/stoc/Ambainis00}
Andris Ambainis.
\newblock Quantum lower bounds by quantum arguments.
\newblock In F.~Frances Yao and Eugene~M. Luks, editors, {\em Proceedings of
  the Thirty-Second Annual {ACM} Symposium on Theory of Computing, May 21-23,
  2000, Portland, OR, {USA}}, pages 636--643. {ACM}, 2000.
\newblock \href {https://doi.org/10.1145/335305.335394}
  {\path{doi:10.1145/335305.335394}}.

\bibitem{barenco1995elementary}
Adriano Barenco, Charles~H Bennett, Richard Cleve, David~P DiVincenzo, Norman
  Margolus, Peter Shor, Tycho Sleator, John~A Smolin, and Harald Weinfurter.
\newblock Elementary gates for quantum computation.
\newblock {\em Physical review A}, 52(5):3457, 1995.

\bibitem{DBLP:journals/corr/abs-2203-05599}
Harry Buhrman, Bruno Loff, Subhasree Patro, and Florian Speelman.
\newblock Memory compression with quantum random-access gates.
\newblock {\em CoRR}, abs/2203.05599, 2022.
\newblock \href {http://arxiv.org/abs/2203.05599} {\path{arXiv:2203.05599}},
  \href {https://doi.org/10.48550/arXiv.2203.05599}
  {\path{doi:10.48550/arXiv.2203.05599}}.

\bibitem{burrows1994block}
Michael Burrows and David Wheeler.
\newblock A block-sorting lossless data compression algorithm.
\newblock In {\em Digital SRC Research Report}. Citeseer, 1994.

\bibitem{DBLP:conf/swat/CleveIGNTTY12}
Richard Cleve, Kazuo Iwama, Fran{\c{c}}ois~Le Gall, Harumichi Nishimura,
  Seiichiro Tani, Junichi Teruyama, and Shigeru Yamashita.
\newblock Reconstructing strings from substrings with quantum queries.
\newblock In Fedor~V. Fomin and Petteri Kaski, editors, {\em Algorithm Theory -
  {SWAT} 2012 - 13th Scandinavian Symposium and Workshops, Helsinki, Finland,
  July 4-6, 2012. Proceedings}, volume 7357 of {\em Lecture Notes in Computer
  Science}, pages 388--397. Springer, 2012.
\newblock \href {https://doi.org/10.1007/978-3-642-31155-0\_34}
  {\path{doi:10.1007/978-3-642-31155-0\_34}}.

\bibitem{darbari}
Parisa Darbari, Daniel Gibney, and Sharma~V. Thankachan.
\newblock Quantum time complexity and algorithms for pattern matching on
  labeled graphs.
\newblock {\em String Processing and Information Retrieval - 28th International
  Symposium, {SPIRE} 2022}, 2022.

\bibitem{DBLP:journals/corr/quant-ph-9607014}
Christoph D{\"{u}}rr and Peter H{\o}yer.
\newblock A quantum algorithm for finding the minimum.
\newblock {\em CoRR}, quant-ph/9607014, 1996.
\newblock URL: \url{http://arxiv.org/abs/quant-ph/9607014}.

\bibitem{equi2023bit}
Massimo Equi, Arianne~Meijer van~de Griend, and Veli M{\"a}kinen.
\newblock From bit-parallelism to quantum string matching for labelled graphs.
\newblock {\em arXiv preprint arXiv:2302.02848}, 2023.

\bibitem{DBLP:conf/focs/Farach97}
Martin Farach.
\newblock Optimal suffix tree construction with large alphabets.
\newblock In {\em 38th Annual Symposium on Foundations of Computer Science,
  {FOCS} '97, Miami Beach, Florida, USA, October 19-22, 1997}, pages 137--143.
  {IEEE} Computer Society, 1997.
\newblock \href {https://doi.org/10.1109/SFCS.1997.646102}
  {\path{doi:10.1109/SFCS.1997.646102}}.

\bibitem{DBLP:conf/focs/FerraginaM00}
Paolo Ferragina and Giovanni Manzini.
\newblock Opportunistic data structures with applications.
\newblock In {\em 41st Annual Symposium on Foundations of Computer Science,
  {FOCS} 2000, 12-14 November 2000, Redondo Beach, California, {USA}}, pages
  390--398. {IEEE} Computer Society, 2000.
\newblock \href {https://doi.org/10.1109/SFCS.2000.892127}
  {\path{doi:10.1109/SFCS.2000.892127}}.

\bibitem{shiwangshiwang}
Yuval Filmus.
\newblock To find median of $k$ sorted arrays of $n$ elements each in less than
  $o(nk\log k)$, Nov 2020.
\newblock URL:
  \url{https://cs.stackexchange.com/questions/87695/to-find-median-of-k-sorted-arrays-of-n-elements-each-in-less-than-onk-log/156925#156925}.

\bibitem{DBLP:conf/soda/GagieNP18}
Travis Gagie, Gonzalo Navarro, and Nicola Prezza.
\newblock Optimal-time text indexing in bwt-runs bounded space.
\newblock In Artur Czumaj, editor, {\em Proceedings of the Twenty-Ninth Annual
  {ACM-SIAM} Symposium on Discrete Algorithms, {SODA} 2018, New Orleans, LA,
  USA, January 7-10, 2018}, pages 1459--1477. {SIAM}, 2018.
\newblock \href {https://doi.org/10.1137/1.9781611975031.96}
  {\path{doi:10.1137/1.9781611975031.96}}.

\bibitem{DBLP:journals/jacm/GagieNP20}
Travis Gagie, Gonzalo Navarro, and Nicola Prezza.
\newblock Fully functional suffix trees and optimal text searching in bwt-runs
  bounded space.
\newblock {\em J. {ACM}}, 67(1):2:1--2:54, 2020.
\newblock \href {https://doi.org/10.1145/3375890} {\path{doi:10.1145/3375890}}.

\bibitem{DBLP:conf/innovations/GallS22}
Fran{\c{c}}ois~Le Gall and Saeed Seddighin.
\newblock Quantum meets fine-grained complexity: Sublinear time quantum
  algorithms for string problems.
\newblock In Mark Braverman, editor, {\em 13th Innovations in Theoretical
  Computer Science Conference, {ITCS} 2022, January 31 - February 3, 2022,
  Berkeley, CA, {USA}}, volume 215 of {\em LIPIcs}, pages 97:1--97:23. Schloss
  Dagstuhl - Leibniz-Zentrum f{\"{u}}r Informatik, 2022.
\newblock \href {https://doi.org/10.4230/LIPIcs.ITCS.2022.97}
  {\path{doi:10.4230/LIPIcs.ITCS.2022.97}}.

\bibitem{DBLP:conf/stoc/Grover96}
Lov~K. Grover.
\newblock A fast quantum mechanical algorithm for database search.
\newblock In Gary~L. Miller, editor, {\em Proceedings of the Twenty-Eighth
  Annual {ACM} Symposium on the Theory of Computing, Philadelphia,
  Pennsylvania, USA, May 22-24, 1996}, pages 212--219. {ACM}, 1996.
\newblock \href {https://doi.org/10.1145/237814.237866}
  {\path{doi:10.1145/237814.237866}}.

\bibitem{DBLP:journals/jda/HariharanV03}
Ramesh Hariharan and V.~Vinay.
\newblock String matching in {\~{o}}(sqrt(n)+sqrt(m)) quantum time.
\newblock {\em J. Discrete Algorithms}, 1(1):103--110, 2003.
\newblock \href {https://doi.org/10.1016/S1570-8667(03)00010-8}
  {\path{doi:10.1016/S1570-8667(03)00010-8}}.

\bibitem{DBLP:journals/eatcs/HoyerS05}
Peter H{\o}yer and Robert Spalek.
\newblock Lower bounds on quantum query complexity.
\newblock {\em Bull. {EATCS}}, 87:78--103, 2005.

\bibitem{DBLP:journals/corr/abs-2211-15945}
Ce~Jin and Jakob Nogler.
\newblock Quantum speed-ups for string synchronizing sets, longest common
  substring, and k-mismatch matching.
\newblock {\em CoRR}, abs/2211.15945, 2022.
\newblock \href {http://arxiv.org/abs/2211.15945} {\path{arXiv:2211.15945}},
  \href {https://doi.org/10.48550/arXiv.2211.15945}
  {\path{doi:10.48550/arXiv.2211.15945}}.

\bibitem{DBLP:conf/stacs/KarkkainenKNPS17}
Juha K{\"{a}}rkk{\"{a}}inen, Dominik Kempa, Yuto Nakashima, Simon~J. Puglisi,
  and Arseny~M. Shur.
\newblock On the size of lempel-ziv and lyndon factorizations.
\newblock In Heribert Vollmer and Brigitte Vall{\'{e}}e, editors, {\em 34th
  Symposium on Theoretical Aspects of Computer Science, {STACS} 2017, March
  8-11, 2017, Hannover, Germany}, volume~66 of {\em LIPIcs}, pages 45:1--45:13.
  Schloss Dagstuhl - Leibniz-Zentrum f{\"{u}}r Informatik, 2017.
\newblock \href {https://doi.org/10.4230/LIPIcs.STACS.2017.45}
  {\path{doi:10.4230/LIPIcs.STACS.2017.45}}.

\bibitem{DBLP:journals/cacm/KempaK22}
Dominik Kempa and Tomasz Kociumaka.
\newblock Resolution of the burrows-wheeler transform conjecture.
\newblock {\em Commun. {ACM}}, 65(6):91--98, 2022.
\newblock \href {https://doi.org/10.1145/3531445} {\path{doi:10.1145/3531445}}.

\bibitem{DBLP:conf/dcc/KempaK17}
Dominik Kempa and Dmitry Kosolobov.
\newblock Lz-end parsing in compressed space.
\newblock In Ali Bilgin, Michael~W. Marcellin, Joan Serra{-}Sagrist{\`{a}}, and
  James~A. Storer, editors, {\em 2017 Data Compression Conference, {DCC} 2017,
  Snowbird, UT, USA, April 4-7, 2017}, pages 350--359. {IEEE}, 2017.
\newblock \href {https://doi.org/10.1109/DCC.2017.73}
  {\path{doi:10.1109/DCC.2017.73}}.

\bibitem{DBLP:conf/esa/KempaK17}
Dominik Kempa and Dmitry Kosolobov.
\newblock Lz-end parsing in linear time.
\newblock In Kirk Pruhs and Christian Sohler, editors, {\em 25th Annual
  European Symposium on Algorithms, {ESA} 2017, September 4-6, 2017, Vienna,
  Austria}, volume~87 of {\em LIPIcs}, pages 53:1--53:14. Schloss Dagstuhl -
  Leibniz-Zentrum f{\"{u}}r Informatik, 2017.
\newblock \href {https://doi.org/10.4230/LIPIcs.ESA.2017.53}
  {\path{doi:10.4230/LIPIcs.ESA.2017.53}}.

\bibitem{kempa2021fast}
Dominik Kempa and Ben Langmead.
\newblock Fast and space-efficient construction of avl grammars from the lz77
  parsing.
\newblock {\em arXiv preprint arXiv:2105.11052}, 2021.

\bibitem{DBLP:conf/stoc/KempaP18}
Dominik Kempa and Nicola Prezza.
\newblock At the roots of dictionary compression: string attractors.
\newblock In Ilias Diakonikolas, David Kempe, and Monika Henzinger, editors,
  {\em Proceedings of the 50th Annual {ACM} {SIGACT} Symposium on Theory of
  Computing, {STOC} 2018, Los Angeles, CA, USA, June 25-29, 2018}, pages
  827--840. {ACM}, 2018.
\newblock \href {https://doi.org/10.1145/3188745.3188814}
  {\path{doi:10.1145/3188745.3188814}}.

\bibitem{DBLP:conf/soda/KempaS22}
Dominik Kempa and Barna Saha.
\newblock An upper bound and linear-space queries on the lz-end parsing.
\newblock In Joseph~(Seffi) Naor and Niv Buchbinder, editors, {\em Proceedings
  of the 2022 {ACM-SIAM} Symposium on Discrete Algorithms, {SODA} 2022, Virtual
  Conference / Alexandria, VA, USA, January 9 - 12, 2022}, pages 2847--2866.
  {SIAM}, 2022.
\newblock \href {https://doi.org/10.1137/1.9781611977073.111}
  {\path{doi:10.1137/1.9781611977073.111}}.

\bibitem{khadiev2022quantum}
Kamil Khadiev, Artem Ilikaev, and Jevgenijs Vihrovs.
\newblock Quantum algorithms for some strings problems based on quantum string
  comparator.
\newblock {\em Mathematics}, 10(3):377, 2022.

\bibitem{DBLP:journals/nc/KhadievR21}
Kamil Khadiev and Vladislav Remidovskii.
\newblock Classical and quantum algorithms for constructing text from
  dictionary problem.
\newblock {\em Nat. Comput.}, 20(4):713--724, 2021.
\newblock \href {https://doi.org/10.1007/s11047-021-09863-1}
  {\path{doi:10.1007/s11047-021-09863-1}}.

\bibitem{DBLP:journals/tit/KiefferY00}
John~C. Kieffer and En{-}Hui Yang.
\newblock Grammar-based codes: {A} new class of universal lossless source
  codes.
\newblock {\em {IEEE} Trans. Inf. Theory}, 46(3):737--754, 2000.
\newblock \href {https://doi.org/10.1109/18.841160}
  {\path{doi:10.1109/18.841160}}.

\bibitem{DBLP:conf/stacs/Kothari14}
Robin Kothari.
\newblock An optimal quantum algorithm for the oracle identification problem.
\newblock In Ernst~W. Mayr and Natacha Portier, editors, {\em 31st
  International Symposium on Theoretical Aspects of Computer Science {(STACS}
  2014), {STACS} 2014, March 5-8, 2014, Lyon, France}, volume~25 of {\em
  LIPIcs}, pages 482--493. Schloss Dagstuhl - Leibniz-Zentrum f{\"{u}}r
  Informatik, 2014.
\newblock \href {https://doi.org/10.4230/LIPIcs.STACS.2014.482}
  {\path{doi:10.4230/LIPIcs.STACS.2014.482}}.

\bibitem{DBLP:journals/tcs/KreftN13}
Sebastian Kreft and Gonzalo Navarro.
\newblock On compressing and indexing repetitive sequences.
\newblock {\em Theor. Comput. Sci.}, 483:115--133, 2013.
\newblock \href {https://doi.org/10.1016/j.tcs.2012.02.006}
  {\path{doi:10.1016/j.tcs.2012.02.006}}.

\bibitem{DBLP:journals/ml/Littlestone87}
Nick Littlestone.
\newblock Learning quickly when irrelevant attributes abound: {A} new
  linear-threshold algorithm.
\newblock {\em Mach. Learn.}, 2(4):285--318, 1987.
\newblock \href {https://doi.org/10.1007/BF00116827}
  {\path{doi:10.1007/BF00116827}}.

\bibitem{DBLP:journals/jacm/McCreight76}
Edward~M. McCreight.
\newblock A space-economical suffix tree construction algorithm.
\newblock {\em J. {ACM}}, 23(2):262--272, 1976.
\newblock \href {https://doi.org/10.1145/321941.321946}
  {\path{doi:10.1145/321941.321946}}.

\bibitem{DBLP:journals/csur/Navarro21a}
Gonzalo Navarro.
\newblock Indexing highly repetitive string collections, part {I:}
  repetitiveness measures.
\newblock {\em {ACM} Comput. Surv.}, 54(2):29:1--29:31, 2021.
\newblock \href {https://doi.org/10.1145/3434399} {\path{doi:10.1145/3434399}}.

\bibitem{DBLP:journals/csur/Navarro21}
Gonzalo Navarro.
\newblock Indexing highly repetitive string collections, part {II:} compressed
  indexes.
\newblock {\em {ACM} Comput. Surv.}, 54(2):26:1--26:32, 2021.
\newblock \href {https://doi.org/10.1145/3432999} {\path{doi:10.1145/3432999}}.

\bibitem{DBLP:conf/stoc/NayakW99}
Ashwin Nayak and Felix Wu.
\newblock The quantum query complexity of approximating the median and related
  statistics.
\newblock In Jeffrey~Scott Vitter, Lawrence~L. Larmore, and Frank~Thomson
  Leighton, editors, {\em Proceedings of the Thirty-First Annual {ACM}
  Symposium on Theory of Computing, May 1-4, 1999, Atlanta, Georgia, {USA}},
  pages 384--393. {ACM}, 1999.
\newblock \href {https://doi.org/10.1145/301250.301349}
  {\path{doi:10.1145/301250.301349}}.

\bibitem{DBLP:conf/mfcs/NishimotoIIBT16}
Takaaki Nishimoto, Tomohiro I, Shunsuke Inenaga, Hideo Bannai, and Masayuki
  Takeda.
\newblock Fully dynamic data structure for {LCE} queries in compressed space.
\newblock In Piotr Faliszewski, Anca Muscholl, and Rolf Niedermeier, editors,
  {\em 41st International Symposium on Mathematical Foundations of Computer
  Science, {MFCS} 2016, August 22-26, 2016 - Krak{\'{o}}w, Poland}, volume~58
  of {\em LIPIcs}, pages 72:1--72:15. Schloss Dagstuhl - Leibniz-Zentrum
  f{\"{u}}r Informatik, 2016.
\newblock \href {https://doi.org/10.4230/LIPIcs.MFCS.2016.72}
  {\path{doi:10.4230/LIPIcs.MFCS.2016.72}}.

\bibitem{DBLP:conf/stoc/Paturi92}
Ramamohan Paturi.
\newblock On the degree of polynomials that approximate symmetric boolean
  functions (preliminary version).
\newblock In S.~Rao Kosaraju, Mike Fellows, Avi Wigderson, and John~A. Ellis,
  editors, {\em Proceedings of the 24th Annual {ACM} Symposium on Theory of
  Computing, May 4-6, 1992, Victoria, British Columbia, Canada}, pages
  468--474. {ACM}, 1992.
\newblock \href {https://doi.org/10.1145/129712.129758}
  {\path{doi:10.1145/129712.129758}}.

\bibitem{DBLP:journals/algorithmica/RaskhodnikovaRRS13}
Sofya Raskhodnikova, Dana Ron, Ronitt Rubinfeld, and Adam~D. Smith.
\newblock Sublinear algorithms for approximating string compressibility.
\newblock {\em Algorithmica}, 65(3):685--709, 2013.
\newblock \href {https://doi.org/10.1007/s00453-012-9618-6}
  {\path{doi:10.1007/s00453-012-9618-6}}.

\bibitem{DBLP:journals/jacm/RodehPE81}
Michael Rodeh, Vaughan~R. Pratt, and Shimon Even.
\newblock Linear algorithm for data compression via string matching.
\newblock {\em J. {ACM}}, 28(1):16--24, 1981.
\newblock \href {https://doi.org/10.1145/322234.322237}
  {\path{doi:10.1145/322234.322237}}.

\bibitem{DBLP:journals/tcs/Rytter03}
Wojciech Rytter.
\newblock Application of lempel-ziv factorization to the approximation of
  grammar-based compression.
\newblock {\em Theor. Comput. Sci.}, 302(1-3):211--222, 2003.
\newblock \href {https://doi.org/10.1016/S0304-3975(02)00777-6}
  {\path{doi:10.1016/S0304-3975(02)00777-6}}.

\bibitem{DBLP:journals/jacm/StorerS82}
James~A. Storer and Thomas~G. Szymanski.
\newblock Data compression via textual substitution.
\newblock {\em J. {ACM}}, 29(4):928--951, 1982.
\newblock \href {https://doi.org/10.1145/322344.322346}
  {\path{doi:10.1145/322344.322346}}.

\bibitem{DBLP:conf/cpm/UrabeNIBT19}
Yuki Urabe, Yuto Nakashima, Shunsuke Inenaga, Hideo Bannai, and Masayuki
  Takeda.
\newblock On the size of overlapping lempel-ziv and lyndon factorizations.
\newblock In Nadia Pisanti and Solon~P. Pissis, editors, {\em 30th Annual
  Symposium on Combinatorial Pattern Matching, {CPM} 2019, June 18-20, 2019,
  Pisa, Italy}, volume 128 of {\em LIPIcs}, pages 29:1--29:11. Schloss Dagstuhl
  - Leibniz-Zentrum f{\"{u}}r Informatik, 2019.
\newblock \href {https://doi.org/10.4230/LIPIcs.CPM.2019.29}
  {\path{doi:10.4230/LIPIcs.CPM.2019.29}}.

\bibitem{valiant2019cs265}
Gregory Valiant.
\newblock Cs265/cme309: Randomized algorithms and probabilistic analysis
  lecture\# 1: Computational models, and the schwartz-zippel randomized
  polynomial identity test.
\newblock 2019.

\bibitem{DBLP:journals/tit/ZivL77}
Jacob Ziv and Abraham Lempel.
\newblock A universal algorithm for sequential data compression.
\newblock {\em {IEEE} Trans. Inf. Theory}, 23(3):337--343, 1977.
\newblock \href {https://doi.org/10.1109/TIT.1977.1055714}
  {\path{doi:10.1109/TIT.1977.1055714}}.

\end{thebibliography}

\appendix

\section{Proof of Lemma \ref{lem:LZ-End-tau}}
\label{appendix:lz-end_proof}

We recite the proof in \cite{DBLP:conf/soda/KempaS22} mostly verbatim (in italics), making any changes in bold. After each portion of the proof, we address any necessary modifications for LZ-End+$\tau$.

We assume $T[n] = \$$. 
\vspace{1em}

\textit{
Let $T^\infty$ be an infinite string defined so that $T^\infty = T[1 + (i - 1) \mod n]$ for $i \in \mathbb{Z}$; in particular,
$T^\infty[1..n] = T[1..n]$. For any $m \geq 1$, let $\mathcal{S}_m = \{S \in \Sigma^m : S \text{ is a substring of } T^\infty\}$. 
Observe that it holds $|\mathcal{S}^m| \leq mz$, since every substring of length $m$ of $T^\infty$ has an occurrence overlapping or starting at some phrase
boundary of the LZ77 factorization in $T$ (this includes the substrings overlapping two copies 
of $T$, since $T[n] = \$$.)
}

\vspace{1em}
No modifications are necessary.
\vspace{1em}

\textit{
For any $j \in [1..z_{e+\tau}]$, let $e_j$ and $\ell_j$ denote (respectively) the last position and the length of the $j^{th}$ phrase in the \textbf{LZ-End+$\tau$} factorization of $T$. Letting $e_0 = 0$, we have $\ell_j = e_j - e_{j-1}$. We call the $j^{th}$ phrase $T(e_{j-1}..e_j]$ special if $j = z_{e+\tau}$, or $j \in [2..z_{e+\tau})$ and $\ell_j \geq \lceil \frac{\ell_{j-1}}{2} \rceil$. Let $z_{e+\tau}'$ denote the number of special phrases. Observe that if $T(e_{j-1}..e_j]$ is not special for every $j \in [i..i+k]$, then $\ell_{i+k} \leq \frac{n}{2^k}$. Thus, any subsequence of $\lfloor \log n \rfloor + 2$ consecutive phrases contains a special phrase, and hence $z_{e+\tau} \leq z_{e+\tau}'(\lfloor \log n \rfloor + 2) = O(z_{e+\tau}' \log n)$.
}

\vspace{1em}
The only changes above are notational, making the proof refer to our modified factorization scheme. It holds be the same logic as with LZ-End that any subsequence of  $\lfloor \log n \rfloor + 2$ consecutive phrases must contain a special phrase.
\vspace{1em}

\textit{
The basic idea of the proof is as follows. With each special phrase of length $\ell$ we associate $\ell$ distinct substrings of length $2^k$, where $2^k < 12\ell$. We then show that every substring $X \in \mathcal{S}_{2^k}$, where $k \in [0..\lceil \log n \rceil + 4]$, is associated with at most two phrases. Thus, by $|\mathcal{S}_{2^k}| \leq z 2^k$, there are no more than $24z$ speciall phrases associated with strings in $\mathcal{S}_{2^k}$. Accounting all $k \in [0..\lfloor \log n \rfloor + 4]$, this implies $z_{e+\tau}' = O(z\log n)$ and hence $z_{e+\tau} = O(z\log^2 n)$.
}

\vspace{1em}
Again the only changes are making $z_e$ into $z_{e + \tau}$, as the approach remains the same.
\vspace{1em}

\textit{
The assignment of substrings is done as follows. Let $j \in [1..z_{e+\tau}]$ be such that $T(e_{j-1}..e_j]$ is a special phrase. Let $k \in [0..\lfloor \log n \rfloor + 4]$ be the smallest integer such that $2^k \geq 6\ell$. With phrase $T(e_{j-1}..e_j]$ we associate substrings $X_i \coloneqq T^\infty(i-2^{k-1}..i+2^{k-1}]$ where $i \in [e_{j-1}..e_j)$. We need to prove two things. First, that every substring is associated with at most two phrases, and second, that for every phrase, all associated $\ell$ substrings are distinct.
}

\vspace{1em}
Again the only changes are making $z_e$ into $z_{e + \tau}$, as the assignment of substrings to special phrases is the same.
\vspace{1em}

\textit{
To show the first claim (that every substring is associated with at most two phrases), suppose that for some $k \in [0..\lfloor \log n \rfloor + 4]$, there exists $X \in \mathcal{S}_{2^k}$ associated with at least three special phrases. Let $T(e_{i-1}..e_i]$ be the leftmost, and $T(e_j..e_j]$ the rightmost such phrase in $T$. Note that since $T[n] = \$$ occurs in $T$ only once, $X$ cannot contain $\$$, and hence we have $j > 1$ and $\ell_j \geq \lceil \frac{\ell_{j-1}}{2} \rceil$. Let $m_i \in [e_{i-1}..e_i)$ and $m_j \in [e_{j-1..e_j})$ be such that $X = T(m_i - 2^{k-1}..m_i+2^{k-1}] = T(m_j - 2^{k-1}..m_j + 2^{k-1}]$ (note that since $X$ does not contain $\$$, the two occurrences of $X$ are inside of $T$, and hence we do not need to write $T^\infty$). Denote $\delta = e_i - m_i$. Observe that $0 < \delta \leq 2^{k-1}$ (since $\delta \leq e_i - e_{i-1} \leq |X|/6$). Thus, $T(m_i - 2^{k-1}..m_i+\delta]$ is an occurrence of string $X[1..2^{k-1}+\delta]$ ending at the end of phrase $T(e_{i-1}..e_i]$. Since we assume at least three phrases are associated with $X$, we have $i < j-1$, and hence the phrase $T(e_{j-2}..e_{j-1}]$ is to the right of $T(e_{i-1}..e_i]$. Recall that we have $e_{j-1} - e_{j-2} < 2(e_j - e_{j-1})$. Thus, by $m_j - 2^{k-1} \leq m_j - 3(e_j-e_{j-1}) \leq e_j - 3(e_j - e_{j-1}) \leq e_{j-2}$ and $e_{j-1} \leq m_j$, the occurrence $T(m_j - 2^{k-1}..m_j+\delta] = X[1..2^{k-1} + \delta] = X[1..2^{k-1} + \delta]$ contains the phrase $T(e_{j-2}..e_{j-1}]$. That, however, implies that when the algorithm is adding the phrase $T(e_{j-2}..e_{j-1}]$ to the LZ-End+$\tau$ factorization, the substring $S = T(e_{j-2}..m_j+\delta]$ has an earlier occurrence (as a suffix of $X[1..2^{k-1}+\delta] = T(m_i - 2^{k-1}..m_i+\delta])$ ending at the end of an already existing phrase $T(e_{i-1}..e_i]$. Since $m_j + \delta > e_{j-1}$, this substring is longer than $T(e_{j-2}..e_{j-1}]$ a contradiction.
}

\vspace{1em}
The key observation that makes the proof translate to LZ-End+$\tau$ is that the possibility of phrase ending at the end of an already existing phrase is sufficient for a contradiction to be created if a substring is associated with more than two phrases.
\vspace{1em}

\textit{
To show the second claim, i.e., that for each special phrase $T(e_{j-1}..e_j]$, all associated $\ell$ strings are distinct, suppose that there exist $i, i' \in [e_{j-1}.. e_j)$ such that $i < i'$ and $X_i = X_{i'}$. Note that by the uniqueness of $T[n] = \$$, we again have $j > 1$ and $\ell_j \geq \lceil \frac{\ell_{j-1}}{2} \rceil$. Denote $Y = T(i-2^{k-1}..i'+2^{k-1}]$ (since under the assumption that $X_i = X_{i'}$, the string $Y$ does not contain $\$$, we again do not need to write $T^\infty$). Since $X_i$ is a prefix of $Y$, $X_{i'}$ is a suffix of $Y$, and it holds $X_i = X_{i'}$, it follows that $Y$ has has period $p$, where $p = i' - i < \ell$. By definition of a period, any substring of $Y$ of length $2xp$, where $x \in \mathbb{Z}_{\geq 0}$, is a square (i.e., a string of the form $SS$, where $S \in \Sigma^*$). Consider this the string $Y' \coloneqq T(e_{j-1}-xp..e_{j-1}+xp]$, where $x \lfloor(e_{j-1} - (i-2^{k-1}))/p\rfloor$. Observe that:
\begin{itemize}
    \item By definition of $x$, the position preceding $Y'$ in $T$ satisfies $e_{j-1}-xp \geq (i-2^{k-1})$. On the other hand, $e_{j-1} + ep \leq e_{j-1} + 2^{k-1} - (i - e_{j-1}) \leq i' + 2^{k-1}$. Thus, $Y'$ is a substring of $Y$.
    \item By $p < \ell$ and $2^{k-1} \geq 3\ell$, we have $xp > ((e_{j-1}-(i-2^{k-1}))/p - 1)p = 2^{k-1} - p - (i-e_{j-1}) \geq \ell$.
\end{itemize}
By the above, the string $Y'$ is a square, i.e., $Y' = SS$ and hence $T(e_{j-1}-xp..e_{j-1}] = T(e_{j-1}..e_{j-1}+xp]$. Moreover, it holds $xp > \ell$. This however, contradicts the fact that $T(e_{j-1}..e_j]$ was selected as a phrase in the LZ-End+$\tau$ factorization of $T$, since $T(e_{j-1}..e_{j-1}+xp]$ is a longer substring with a previous occurrence $T(e_{j-1}-xp..e_{j-1}]$ ending at a phrase boundary. Thus, we have shown that all $\ell$ strings associated with $T(e_{j-1}..e_j]$ are distinct.
}

\vspace{1em}
Again, the observation that makes the proof translate from LZ-End to LZ-End+$\tau$ is that the possibility of phrase ending at the end of an already existing phrase is sufficient for a contradiction if all $\ell$ strings associated with a special phrase are not distinct.


%

\section{Threshold Problem for $t = 1$ with a Promise}
\label{sec:threshold}
This is a folklore result. We include it for completeness.

\begin{lemma}
\label{lem:lb_for_check_if_0}
Given the promise that $T[1,n]$ contains either zero $1$'s or exactly one $1$, $\Omega(\sqrt{n})$ input queries are required to solve the Threshold problem with $t = 1$.
\end{lemma}

\begin{proof}
The adversarial lower bound techniques of Ambainis' state the following:
\begin{lemma}[\cite{DBLP:conf/stoc/Ambainis00}]
\label{lem:ambainis}
Let $F(x_1, \dots , x_N)$ be a function of $n$ $\{0, 1\}$-valued variables and $X$, $Y$ be two sets of inputs such
that $F(x) \neq F(y)$ if $x \in X$ and $y \in Y$ . Let $R \subseteq X \times Y$ be such that
\begin{enumerate}
\item For every $x \in X$, there exist at least $m$ different $y \in Y$ such that $(x, y) \in R$.
\item For every $y \in Y$ , there exist at least $m'$ different $x \in X$ such that $(x, y) \in R$.
\item For every $x \in X$ and $i \in \{1, \dots , n\}$, there are at most $l$ different $y \in Y$ such that $(x, y) \in R$ and $x_i \neq y_i$.
\item For every $y \in Y$ and $i \in \{1, \dots , n\}$, there are at most $l'$ different $x \in X$ such that $(x, y) \in R$ and $x_i \neq y_i$.
\end{enumerate}
Then, any quantum algorithm computing $f$ uses $\Omega(\sqrt{\frac{mm'}{ll'}})$ queries.
\end{lemma}

Let $F$ map binary strings of length $n$ to $\{0,1\}$ and let $F(T) = 1$ if $T$ has one $1$ and $F = 0$ if $T$ has no $1$s. Consider sets $X = \{0^n\}$ and $Y = \{0^i10^{n-i-1} \mid 1 < i \leq n\}$, and relation $R = \{(x,y) \in X \times Y\}$.  The values appearing in Lemma \ref{lem:ambainis} are $m = |Y| = \Omega(n)$, $m' = |X| = 1$, $l = 1$, and $l' = 1$. We conclude that determining $F$ requires $\Omega(\sqrt{n})$ queries.
\end{proof}

\end{document}